\PassOptionsToPackage{square,comma,numbers,sort&compress,super}{natbib}

\documentclass[12pt]{article}

\usepackage{pdflscape}

\usepackage{geometry}

\usepackage{type1cm}

\usepackage{tikz}

\usepackage{etoolbox}
\makeatletter
\patchcmd{\@sect}{#8}{\boldmath #8}{}{}
\let\ori@chapter\@chapter
\def\@chapter[#1]#2{\ori@chapter[\boldmath#1]{\boldmath#2}}
\makeatother

\usepackage{amsmath}
\usepackage{graphicx,psfrag,epsf}
\usepackage{enumerate}
\usepackage{caption}

\usepackage{mathtools}

\usepackage{bm}

\usepackage{natbib}

\setcitestyle{square, comma, numbers,sort&compress, super}

\usepackage{url}

\usepackage[labelfont=bf]{caption}

\usepackage{threeparttable}

\usepackage{array}

\usepackage{amsthm}

\newtheorem*{newthm}{Theorem}

\newtheorem*{lemma*}{Lemma}

\usepackage{diagbox}

\usepackage[T1]{fontenc}

\usepackage[toc,page]{appendix}

\usepackage{setspace}

\usepackage{epstopdf}
\usepackage[caption=false]{subfig}

\usepackage[numbers,sort&compress]{natbib}
\bibpunct[, ]{[}{]}{,}{n}{,}{,}

\newtheorem*{proposition*}{Proposition}

\theoremstyle{definition}

\usepackage[english]{babel}
\usepackage{amssymb}
\usepackage{textcmds}
\usepackage{hyperref}

\usepackage[english]{cleveref}

\usepackage{xcolor}

\usepackage{graphicx}
\usepackage[section]{placeins}

\usepackage{float}

\usepackage{url}

\usepackage{breakurl}

\usepackage[caption = false]{subfig}
\usepackage{authblk}

\newcommand{\blind}{0}

\usepackage{numprint}
\npdecimalsign{.}
\nprounddigits{4}

\begin{document}

\if0\blind
{

\title{\bf Deriving the Variance-Minimizing Design for Standard Addition via c-Optimality.}

\author[1]{Gerhard G\"ossler\thanks{
    Corresponding author: gerhard.goessler@uni-graz.at}}
    \author[2]{Vera Hofer}
\author[1]{Walter Goessler}

\affil[1]{Institute of Chemistry, Analytical Chemistry, University of Graz, Austria}
\affil[2]{Institute of Operations and Information Systems, University of Graz, Graz, Austria} 

\date{}

  \maketitle
} \fi

\if1\blind
{
  \bigskip
  \bigskip
  \bigskip
  \begin{center}
    {\LARGE\bf Title}
\end{center}
  \medskip
} \fi

\vspace*{-1\baselineskip}

\bigskip
\begin{abstract}

\noindent Knowledge about optimal designs for standard addition seems to be scattered among literature and is also, at least partially, only available in mathematical literature that is not quickly accessible for readers not skilled in the field of design optimality theory. Therefore, the idea for this work was to summarize what is already available in analytical literature and to apply the respective results from optimality theory, where needed, to the special case of standard addition. It is shown, for measurement errors that are non-decreasing, e.g., are constant or increase  linearly or quadratically with increasing analyte concentration, that the optimal design in the case of a linear response is a two-point design irrespective of the particular behavior of measurement error variance. In addition, it is demonstrated that the optimal allocation of measurements depends on the concrete setting, which means that the optimal distribution of measurements may deviate significantly from a 50:50 ratio. It is also investigated how the range, i.e., the largest added concentration influences the result. Last but not least, also the question of applying weighted regression is discussed and it is shown, that, in contrast to designs using more than two spiked concentrations, no weighting is necessary to achieve optimal results, when a two-point design is used. While the focus lies on the precision of the concentration estimate also the implications for the bias are investigated.

\end{abstract}

\noindent%
{\it Keywords:} standard addition, extrapolation approach, trueness and precision, bias and variability, homo- and heteroscedasticity, approximation formulas, constrained optimization, Karush-Kuhn-Tucker conditions, Elfving-theorem, c-optimality, two-point design

\vfill

\newpage

\section{Introduction}


This work aims to provide a comprehensive analysis of design optimality for standard addition. Since the mathematical treatment of this topic presented in the Methods section is rather extensive, readers with a primarily applied focus may wish to proceed as follows: read the Introduction to become familiar with the core concepts and notation used in this work, and then proceed directly to the Discussion section for a concise, practically oriented summary.\\

This work investigates experimental designs, referred to as designs in the following. For extensive information on this topic, see, for example, the classic textbook by Montgomery \cite{montgomery.2017}. The design in the case of standard addition is given by the number of independent measurements $n$ which are taken at $2 \leq n_e \leq n$ different spiked concentrations $x_i \in [0,r], i=1,...,n_e$, with $r > 0$ denoting the chosen upper concentration limit for spiking. At each of the concentrations $x_i$ a certain number $1 \leq n_i < n$ of independent repeated measurements is taken. Since the $n_i$ do not have to be the same for all $i=1,...,n_e$ they can take on every number $\geq 1$ as long as $n=\sum n_i$ holds. Therefore, the design is represented by the vector $\bm{x} \in \mathbb{R}^n$ as follows: $\bm{x}=(x_{11},...,x_{1n_1},x_{21},...,x_{2n_2},...,x_{n_e1},...,x_{n_en_{n_e}})$ with $x_i=x_{i1}=x_{i2}=...=x_{in_i}, i=1,...,n_e$. Thus, from a geometric point of view, $\bm{x}$ contains the coordinates of a certain point in the $n$-dimensional hypercube $[0,r]^n$.
 Furthermore, also the chosen value of $r$ is of significant importance for the performance of standard addition and should be based on a well-founded decision.\\ 

Generally speaking, the performance of an estimator depends on the relationship between the independent and dependent variables, the measurement error, and the functional form of the estimator itself, which determines how these inputs are processed and how errors propagate from measurement to estimate. For standard addition, the assumed relationship between the instrumental response and the analyte concentration is linear and superimposed by normally distributed measurement errors. A comprehensive discussion of the homoscedastic case is provided in \cite{Gossler.2025}. The homoscedastic case can be easily extended to the heteroscedastic case, which is very common for many different instrumental techniques (see e.g. \cite{Franke.1978} and \cite{Ellison.2008}, but see also \cite{Goncalves.2016} for a case in which the data appear to be homoscedastic), by assuming that the measurement error varies as a function of the analyte concentration. Mathematically this can be expressed as follows: Let $Y$ denote the response, $C_0$ the unknown analyte concentration in the sample, $x$ the added concentration of the analyte, $c=C_0 + x$ the analyte concentration in the spiked sample, $\beta_0$ and $\beta_1$ intercept and slope of the assumed linear relationship, and $\varepsilon(c)$ the term representing the concentration dependent measurement error. Then, since it is assumed that the blank value of the applied analytical method is not significantly different from zero, the relationship between concentration and response is given by

\vspace*{-1\baselineskip}

\begin{equation} \label{eq1}
Y(c) = \beta_1c + \varepsilon(c) = Y(C_0 + x) = \beta_1C_0 + \beta_1x + \varepsilon(c) = \beta_0 + \beta_1x + \varepsilon(c) 
\end{equation}

\vspace*{.25\baselineskip}

with $\beta_0 = \beta_1C_0$, $ \varepsilon(c) \sim N(0,\sigma(c))$ and $\sigma(\cdot)$ indicating the functional relationship between analyte concentration and error variability. Based on this, $C_0$ equals $\beta_0/\beta_1$ and therefore, the proper estimator for $C_0$ is given by

\vspace*{-1\baselineskip}

\begin{equation} \label{eq2}
\hat{C}_0=\hat{\beta}_0/\hat{\beta}_1
\end{equation}

\vspace*{.25\baselineskip}

with $\hat{\beta}_0$ and $\hat{\beta}_1$ being estimators for the unknown intercept and slope of the assumed linear relationship gained by applying linear regression.\\

The performance of linear regression depends on the chosen design $\bm{x}$. The concept of D-optimality \cite{montgomery.2017} provides a criterion for selecting experimental designs that yield efficient estimation of model parameters which, among other applications, has been widely adopted in analytical chemistry - see  \cite{valverde2023} and \cite{wang2025} for some recent examples. Since $\hat{\beta}_0$ as well as $\hat{\beta}_1$ are random quantities due to measurement error also $\hat{C}_0$ is a random quantity. Therefore, the performance of $\hat{C}_0$ is best measured in terms of precision and bias. Hence, it is reasonable to choose the design such that, for a given predetermined number $n$ of independent measurements, the resulting estimator attains the highest possible precision and the lowest possible bias. Throughout this work, the precision of  $\hat{C}_0$ or, equivalently, its variance, serves as the optimality criterion. The objective is to maximize the precision of $\hat{C}_0$, which is equivalent to minimizing its variance. Nevertheless, since the bias should of course also be as small as possible it is also investigated how the optimal design with respect to precision influences this quantity.\\

It will be shown, that the optimal design with respect to the precision of $\hat C_0$, in the following also simply denoted as the \hyperlink{start}{\hypertarget{vod}{\textit{variance optimal design}}}, is, under all circumstances considered, a two-point design, i.e., a design with $x_1=0$ and only one added concentration $x_2$ ($n_e=2$). It will be further demonstrated how to determine the optimal added concentration and how to optimally allocate the chosen number of measurements $n$ to  $x_1$ and $x_2$. As will become clear in the following sections, the optimal number of measurements for $x_1$, denoted $n_1$, does not necessarily coincide with the optimal number of measurements for $x_2$, denoted $n_2$. Furthermore, depending on the behavior of the variance, $x_2$ does not always correspond to the upper limit of the permissible concentration range; it may also be below it. To ensure a clear distinction, in the following all designs that use more than one added concentration are consistently referred to as multi-point designs.\\ 

When considering heteroscedasticity, another significant advantage of two-point designs emerges, making their application even more compelling: unlike multi-point designs, weighted regression \cite{Aitken.1935} is unnecessary for two-point designs, which avoids the (potentially problematic) determination of weights. On the other hand, weighting is strictly recommended when using multi-point designs. Even in the case of data indicating homoscedastic errors it could still make sense to assume heteroscedasticity if such behavior is known for the applied method. Simulation results show  that, in cases where homoscedasticity appears plausible based on the data, ignoring a true heteroscedastic behavior of the applied method could  worsen the result significantly in cases where a multi-point design is applied.\\

Determining $\bm{x} \in [0,r]^n$ which minimizes $\sigma^2_{\hat{C}_0}$ subject to $0 \leq x_i \leq r$ for all $i=1,...,n$ is a nontrivial optimization task. One way to address this problem is to apply the theory of Karush, Kuhn, and Tucker \cite{karush.1939}, which generalizes the method of Lagrange multipliers \cite{nocedal.2006} to optimization problems with inequality constraints refered to as the KKT conditions. However, this approach depends on the explicit functional form of the objective function and, therefore, the entire calculation must be repeated whenever the variance structure, and thus the objective function, changes.  Alternatively, the problem can be embedded into the framework of optimal experimental design as developed, for example, in the monograph of Pukelsheim \cite{pukelsheim.2006}. Using results such as Elfving's theorem \cite{elfving.1952} and  Carathéodory's theorem \cite{tyrrell.1970}, it is possible to derive structural properties of optimal designs that hold independently of the particular variance of the observations. In particular, these results imply that an optimal design is supported on only a small number of uniquely determined points. A problem with determining the optimal design is that this requires knowledge of the underlying parameters, which is of course not the case in practice. Fortunately, however, it will become apparent that two-point designs using a 50:50 allocation of the measurements are, at least if the variance increases at most quadratically, always significantly better than common multi-point designs, e.g., designs with $n_e \in \{3,4,5\}$,  $x_i$ evenly distributed over $[0,r]$ and $n_1=...=n_{n_e}$ for all $i=1,...,n_e$.\\

The question of optimal designs for standard addition seems to be scarcely addressed in the analytical literature. A very recommendable and accessible overview is given by Stephen et al. \cite{Ellison.2008}. Two examples of the literature devoted to the mathematical treatment of the question of optimal designs for standard addition are the publications by Franke et al. \cite{Franke.1978} and Ratzlaff \cite{ratzlaff.1979}. There are numerous publications discussing optimality of experimental designs in the mathematical literature on optimality theory. Optimal experimental design (e.g., the work of Elfving  \cite{elfving.1952} or Pukelsheim \cite{pukelsheim.2006}) is mathematically grounded in \textit{convex analysis} and \textit{convex geometry}. Only a rather small number of publications apply this mathematical optimality theory directly to concrete analytical problems. One example is the work of Kitsos and Kolovos, see e.g., \cite{kitsos.2010}, who use optimal design theory to study the calibration of a pH meter. However, to our knowledge, these approaches have not yet been explicitly applied to the problem of standard addition. Moreover, many contributions based on optimal design theory are mathematically involved to a degree that makes them rather difficult to access.\\ 

Therefore, the aim of this work is to extract all necessary theoretical results from the available literature, as summarized in the Methods  section, in order to explicitly solve the problem of optimal designs for standard addition based on optimality theory. Furthermore, the solution to the problem is presented in a mathematically rigorous manner, while at the same time allowing readers who are primarily interested in the applications rather than the underlying theory to easily find the information relevant to their needs. To this end, the Results and Discussion sections present and analyze, in a transparent and comprehensive manner, the outcomes of calculations and simulations based on Elfving’s theorem and its geometric characterization of the optimal design problem. The results presented comprise the determination of optimal designs and also the comparison of the performance of optimal, 50:50 two-point and common multi-point designs with respect to bias and precision for several distinct analytical settings.  These settings vary with respect to variance behavior and range. This allows the effects on the optimal design, in terms of the optimal value for $x_2$ and the corresponding allocation of measurements to $x_1$  and $x_2$, to be examined and enables a comparison of the performance of the aforementioned designs. Finally, a heuristic approach is proposed for practical implementation, aiming to achieve near-optimal performance in the absence of exact knowledge of the optimal design.

\section{Methods}

\subsection{Applying optimality theory: Preliminaries}

\paragraph{Problem formulation}\mbox{}\\

Consider the linear regression model presented in \ref{eq1}, i.e., $Y(x)=\beta_0+\beta_1 x+\varepsilon(x), x \in \mathcal{X} \ \text{and} \  \varepsilon(x) \sim N(0,\sigma(x))$ where the \textit{design variable} $x$ represents the added concentration restricted to the compact \textit{design region} $\mathcal{X}=[0,r]$ and $ \sigma(x)$ denotes the standard deviation of the measurement error which might be constant or increases with increasing $x$. Note that in the heteroscedastic case, the error terms are typically still assumed to have mean zero and to be uncorrelated, but their variances are not constant. The ordinary least squares (OLS) estimators for estimating $\beta_0$ and $\beta_1$ remain unbiased and consistent in this case, but are no longer efficient, and the usual variance formulas and inference procedures are invalid. If the variance structure is known (up to a multiplicative constant), weighted least squares (WLS) \cite{Montgomery.2013} provides efficient estimators. In particular, WLS yields the best linear unbiased estimators (BLUE \cite{rao.1973}) and leads to valid inference when the weights are correctly specified. Compared to OLS, it avoids misleading hypothesis tests and confidence intervals and it reduces the variance of the estimators within the class of linear unbiased estimators.\\

The weights  $w_i, \ i=1,...,n_e$, needed for WLS, should be chosen to be inversely proportional to the variances $\sigma_i^2=\sigma^2(x_i)$ of the responses $Y_i=Y(x_i)$, i.e., $w_i \propto 1/\sigma_i^2$. Multiplying the regression equation with the square root of the chosen weights, $\sqrt{w_i}$, yields the following transformed regression equation:

\vspace*{-.5\baselineskip}

\[
\sqrt{w_i}Y_i = \sqrt{w_i}\beta_0 + \sqrt{w_i}\beta_1x_i + \sqrt{w_i}\varepsilon_i  = \sqrt{w_i}\beta_0 + \sqrt{w_i}\beta_1x_i + \varepsilon^w_i= Y^w_i.
\]

\vspace*{.25\baselineskip}

By choosing a proportionality factor $\sigma^2$ \footnote{not to be confused with the function $\sigma(\cdot)$}, the variances can, for all $i=1,...,n_e$,  be factorized as $\sigma_i^2=\sigma^2v_i \ \Rightarrow \ v_i = \sigma^2_i/\sigma^2$ which defines a variance function $v(x) \geq 0$ by $v(x_i):=v_i$. Choosing $w_i=1/v_i$ results in 

\vspace*{-1\baselineskip}

\begin{equation} \label{propf}
var(Y^w_i )=var(\varepsilon^w_i)=var(\sqrt{w_i}\varepsilon_i)=w_i\sigma_i^2=\sigma_i^2/v_i=\sigma^2 
\end{equation}

\vspace*{.25\baselineskip}

\noindent for all $i=1,...,n_e$. It is thus not necessary to know the actual variances, rather it is sufficient to know the variances up to a proportionality factor, i.e., applying these weights leads to transformed observations $Y^w_i$, which obey the assumption of homoscedasticity. Therefore, the transformed regression equations fulfill the requirements of simple linear regression and can be solved by simply applying the approach of ordinary least squares which yields the parameter estimates $\hat{\beta}_0^w$ and $\hat{\beta}_1^w$ by minimizing the weighted sum of squared residuals, i.e., 

\vspace*{-.5\baselineskip}

\[
(\hat{\beta}_0^w , \hat{\beta}_1^w)=\underset{(\beta_0, \beta_1) \in \mathbb{R}^2}{\text{arg min}}\sum_{i=1}^{n} w_i(y_i - (\beta_0 + \beta_1x_i))^2
\]

\vspace*{.25\baselineskip}

\noindent with $y_i$ denoting an actual observed value of the response $Y_i$. Consequently, the estimator of $C_0$ based on weighted regression is given by 

\vspace*{-1.25\baselineskip}

\[
\hat{C}_0^{w}=\hat{\beta}^w_0/\hat{\beta}^w_1.
\]

\vspace*{.25\baselineskip}

Given the above discussion (see also \cite{Ellison.2008}), it is, for all that follows, therefore assumed that $\sigma(x)=\sigma\sqrt{v(x)}$,  with $\sigma > 0$ and variance function $v(x)$ (with repsect to the variance of the response $Y(x)$) given as

\vspace*{-1.25\baselineskip}

\begin{equation}\label{varfunc}
v(x)=V_0+(\beta_1(C_0+x))^k=V_0+y^k
\end{equation}

\vspace*{.25\baselineskip}

with $0 \leq V_0$, $0 \leq k$, $C_0$ denoting the unknown concentration and $y$ denoting the mean response at $c=C_0+x$. Then, the variance of $Y(x)$ is given by $\sigma^2(x)=\sigma^2 v(x)$ with $\sigma > 0$. Let $\sigma_0$ denote the standard deviation for the blank, then the parameter $V_0$  has to be chosen such that $(\sigma^2V_0)^{0.5}=\sigma_0$. $k=0$ represents the homoscedastic case (e.g., $\sigma^2(x)=\sigma^2$ when setting $V_0=0$)  and if $k>0$ applies, the variance of the response  increases with increasing $x$, e.g., $k=2$ and $V_0=0$ represent the case of constant RSD with RSD$=\sigma$. $V_0 > 0$ is of interest if $k>0$ and it is reasonably assumed that the measurement variability for the blank is not equal to zero.  If $x$ is large enough, $V_0$ can be neglected and the relative standard deviation (RSD) at a certain concentration $c=C_0+x$ becomes approximately:

\vspace*{-1\baselineskip}

\[
RSD(c)=\frac{\sigma(x)}{Y(x)} \approx \frac{\sigma \sqrt{(\beta_1(C_0 + x))^k}}{\beta_1(C_0 + x)}=\sigma (\beta_1(C_0 + x))^{k/2-1}
\]
which yields a constant RSD for $k=2$ since $\sigma (\beta_1(C_0 + x))^{2/2-1}=\sigma (\beta_1(C_0 + x))^{0}=\sigma$.\\

To facilitate further discussion, the regression model is rewritten in vector form as
\[
Y(x)=g(x)^T\bm{\beta}+\varepsilon(x),
\]
where the parameter vector $\bm{\beta}$ and the \textit{regression function} $g(x)$ are defined as follows
\[
\bm{\beta}=
\begin{pmatrix}
\beta_0\\
\beta_1
\end{pmatrix},
\qquad
g(x)=
\begin{pmatrix}
1\\
x
\end{pmatrix}.
\]

Multiplying both sides of the regression equation by $1/\sqrt{v(x)}$ results in a transformed regression equation with homoscedastic errors, i.e.,
\[
\frac{Y(x)}{\sqrt{v(x)}}=\frac{g(x)^T}{\sqrt{v(x)}}\bm{\beta}+\frac{\varepsilon(x)}{\sqrt{v(x)}}=\tilde{g}(x)^T\bm{\beta}+\tilde{\varepsilon}(x)=\tilde{Y}(x)
\]

\[ \text{with} \ \tilde{\varepsilon}(x)\sim N(0,\sigma) \ \text{and} \ 
\tilde g(x)=\frac{1}{\sqrt{v(x)}}
\begin{pmatrix}
1\\
x
\end{pmatrix}
\]

denoting the homoscedastic error term and the transformed regression function of the transformed equation. For the sake of simplicity, the vector valued function $f:\mathcal X \to \mathbb R^2$ defined as 
\[
f(x)=(f_1(x),f_2(x))^T := 
\begin{cases} 
g(x) & \text{if }k= 0 \\
\tilde g(x)   & \text{if } k > 0 
\end{cases}
\]
will be considered in the following. Consequently
\[
Y(x)=f(x)^T\bm{\beta}+ \varepsilon \quad \text{and} \quad \tilde Y(x)=f(x)^T\bm{\beta}+\tilde \varepsilon
\] 
represent the homoscedastic and the transformed heteroscedastic case, respectively.\\

Note, that in the present context of optimal design theory, the term \emph{regression function} refers to the vector-valued mapping $f(\cdot)$ while \emph{regression vector} denotes its value $f(x)$ at a given design point $x \in \mathcal X$. This terminology is specific to the present setting (see, e.g., \cite{pukelsheim.2006}) and should be distinguished from other uses of these terms in the regression literature.


\paragraph{\boldmath Variance and Bias of $\hat C_0$}\mbox{}\\

Since the variance and the bias of the ratio $\hat{C}_0=\hat{\beta}_0/\hat{\beta}_1$ are the quantities of interest, it is necessary to use approximations (see e.g., \cite{Gossler.2025}) of these quantities in order to obtain analytically tractable expressions rather than relying on simulation. Applying the method of propagation of error yields, based on the 1st and 2nd Taylor polynomial of $\hat{C}_0$ (see Supplement for details) the following approximation formulas:

\vspace*{-1\baselineskip}

\begin{equation} \label{varf}
var(\hat{C}_0) \approx \left(\frac{1}{\beta_1}\right)^2\sigma_{{\hat{\beta}}_0}^2+\left(\frac{\beta_0}{\beta_1^2}\right)^2\sigma_{{\hat{\beta}}_1}^2-\frac{2cov\left({\hat{\beta}}_0,{\hat{\beta}}_1\right)\beta_0}{\beta_1^3}
\end{equation}

\begin{equation} \label{biasf}
\begin{aligned}
Bias(\hat{C}_0)=\mathbb{E}(\hat{C}_0-C_0) \approx 
\frac{\beta_{0}}{\beta_{1}^3}\sigma^2_{\hat{\beta}_1} - \frac{1}{\beta_{1}^2}cov(\hat{\beta}_0,\hat{\beta}_1) 
\end{aligned}
\end{equation}

By means of these approximations, the relevant quantities are linearized, thereby making $var(\hat{C}_0)$ accessible to the application of Elfving's theorem as will be shown in the Results section. In contrast, the approximation of $Bias(\hat{C}_0)$ does not allow its application.\\


In the following, expressions are provided for the variances $\sigma^2_{\hat \beta_0}$ and $\sigma^2_{\hat \beta_1}$ appearing in \ref{varf} and \ref{biasf} for properly weighted regression. In the homoscedastic case, these expressions are standard and can be found in any textbook on regression analysis, see e.g., \cite{Montgomery.2013}, but might not be as easily accessible for the heteroscedastic case. Therefore, these formulas are given below for the heteroscedastic case when weighted regression is applied with weights $w_i$ inversely proportional to the actual variances, i.e., $w_i=1/\sigma_i^2$ (details on the derivation of this result can be found in the Supplement):

\begin{equation}\label{covar} 
var(\hat{\bm{\beta}})=\left(\begin{matrix}\sigma_{{\hat{\beta}}_0}^2&cov\left({\hat{\beta}}_0,{\hat{\beta}}_1\right)\\cov\left({\hat{\beta}}_0,{\hat{\beta}}_1\right)&\sigma_{{\hat{\beta}}_1}^2\\\end{matrix}\right)=\frac{1}{\sum{w_i\sum{w_ix_i^2-\left(\sum{w_ix_i}\right)^2}}}\left(\begin{matrix}\sum{w_ix_i^2}&-\sum{w_ix}_i\\-\sum{w_ix}_i&\sum w_i\\\end{matrix}\right)
\end{equation}


\paragraph{Some key concepts from optimal design theory}\mbox{}\\

Trying to find an optimal exact design for a fixed number $n$ of experimental runs typically leads to a difficult optimization problem: the support points and their integer replication numbers must be determined simultaneously. To make this problem more tractable the integrality constraint on the number of replications is relaxed. Instead of working directly with integer counts, proportions of observations assigned to design points are considered. This yields a continuous formulation in which a design is represented by a probability measure on $\mathcal{X}$. The resulting optimization problem is typically convex and admits powerful analytical and geometric tools \cite{pukelsheim.2006}. This motivates the distinction between \textit{exact designs}, which correspond to finite samples with integer replication numbers, and \textit{approximate designs}, which arise from this continuous relaxation. The latter serve both as a theoretical benchmark and as a practical guide for constructing efficient exact designs.\\ 

An exact design $\xi_n$ for sample size $n$ is defined by
\[
\xi_n=\{(x_i,n_i)\}_{i=1}^{n_e} \ \text{with} \ n_i \in \mathbb{N}, 1 \leq n_i \leq n \ \text{and} \ \sum \ n_i=n
\]
where $x_i \in \mathcal{X}$ are the \textit{support points} of the design and the $n_i$ are the numbers of measurements taken at each of these points.\\ 

Each exact design induces a corresponding approximate design $\xi$ \cite{kiefer.1959}. By setting \textit{weights} $\kappa_i = n_i/n$ ($\kappa$ is used for a clear distinction from the symbol for the weights in the context of weighted regression) the induced approximate design is given by

\vspace*{-.4\baselineskip}

\[
\xi=\{(x_i,\kappa_i)\}_{i=1}^{n_e}.
\]

\vspace*{.1\baselineskip}

In general, an approximate design $\xi$ is defined as a probability distribution over $\mathcal{X}$, i.e., each support point $x_i \in \mathcal{X}$ of the design is assigned a weight $\kappa_i$ ($\sum \kappa_i = 1$ and $\kappa_i \ge 0$), which represents the proportion of total observations at point $x_i$ and induces a design specific  weight function $\kappa_{\xi}(\cdot)$ on $\mathcal{X}$ by $\kappa_{\xi}(x):=\kappa_i$ if $x=x_i$ and $\kappa_{\xi}(x):=0$ for $x$ not part of the support of $\xi$. According to \cite{kiefer.1959} "These measures will always have finite support, this being the only type which it is practically meaningful or necessary to consider". This can be justified formally in the case of linear regression by Carathéodory's theorem \cite{tyrrell.1970}, which implies that it suffices to restrict attention to designs with finite support. The set of support points of such designs $\xi_n$ and $\xi$ is given by
\[
\text{supp}(\xi_n) = \{x_1,...,x_{n_e}\}  \ \text{and}  \  \text{supp}(\xi) = \{ x \in \mathcal{X} \mid \kappa_{\xi}(x) > 0 \}
\]

with $|\text{supp}(\xi_n)|  \ \text{and} \ |\text{supp}(\xi)|$ denoting the number of their support points, e.g., $|\text{supp}(\xi_n)| =n_e$.\\

To compare and optimize designs in a way that is independent of the sample size $n$ the \textit{information matrix} of an approximate design $\xi$ with finite support of size $n_e$ (to include infinite supports an integral notation would be necessary) is defined by

\vspace*{-.5\baselineskip}

\[
M(\xi):=\sum_{i=1}^{n_e} \kappa_i\, f(x_i) f(x_i)^T
\]

with $x_i \in \text{supp}(\xi)$, $f(x)$ denoting the regression function and $\kappa_i=\kappa_{\xi}(x_i)$. For an exact design $\xi_n$ with $n$ observations,  design matrix $X$ and corresponding approximate design $\xi$, one has $X^TX = nM(\xi)$ (see Supplement). Hence, criteria for the quality of an exact design that are based on the covariance structure of the least squares estimator can be expressed, up to a constant factor, solely in terms of $M(\xi)$, i.e.,

\vspace*{-1.5\baselineskip}

\[
var(\hat{\bm{\beta}}) \propto M(\xi)^{-1}
\]

\vspace*{.1\baselineskip}

with covariance matrix $var(\bm{\hat \beta})$ as given in equation \ref{covar}. Approximate designs therefore replace the exact design problem, which typically involves both the selection of support points and their integer allocations, by a continuous optimization problem over probability measures, with $M(\xi)$ capturing all relevant design information.\\

For practical implementation, i.e., when converting an approximate design into an exact design, the integer replication numbers $n_i$ corresponding to the weights $\kappa_i$ and the total sample size $n$ are obtained by setting $n_i \approx n \kappa_i$ and applying a rounding procedure such that $\sum_i n_i = n$. For example, a so called efficient rounding procedure (see, e.g., \cite{pukelsheim.2006}) transforms an approximate design into an exact design of size $n$ by assigning integer replication numbers $n_i$ such that the resulting information matrix is as close as possible to that of the approximate design, thereby minimizing the loss in the chosen optimality criterion. In addition, constraints such as $n_i \geq 1$ for all support points may be imposed to ensure a nonsingular information matrix and hence estimability of the model parameters.\\


Let $\xi_n$ be an exact design with $n$ observations used to estimate $\bm \beta=(\beta_1,...,\beta_k)^T$. For this design, let $\hat{\bm \beta}$ denote the least squares estimator of $\bm \beta$. Then, for arbitrary vectors $\bm c \in \mathbb{R}^k$, the estimator $\bm c^T \hat{\bm \beta}$ is the BLUE of $\bm c^T \bm \beta$, and its variance satisfies
\[
var(\bm c^T \hat{\bm \beta})=\bm c^T var(\hat{\bm \beta}) \bm c = \frac{\sigma^2}{n} \, \bm c^T M(\xi)^{-1} \bm c
\]

with $\xi$ denoting the corresponding approximate design of $\xi_n$ and $\sigma^2$ denoting a proportionality constant. Thus, up to the constant factor $\sigma^2/n$, the variance of $\bm c^T \hat{\bm \beta}$ depends on the design only through the information matrix $M(\xi)$. This suggests that an optimal BLUE with respect to the given design region $\mathcal{X}$ can be attained by an appropriate choice of the design which motivates the definition of $c$-optimality: an approximate design $\xi^*$ on $\mathcal X$ is called \textbf{$\bm c$-optimal} if

\[
\xi^* \in \arg\min_{\xi} \; \bm c^T M(\xi)^{-1} \bm c,
\]

where the minimum is taken over all approximate designs $\xi$ on $\mathcal X$ with nonsingular information matrix. This analytic formulation admits an equivalent geometric characterization which leads to Elfving's theorem.\\

\vspace{-.2cm}

\subsection{Applying optimality theory: Elfving's Theorem}

Since the application of the theory of Karush, Kuhn, and Tucker \cite{karush.1939} does not yield generic solutions for the problem of minimizing the variance of $\hat C_0$, i.e., needs to be applied separately for each specific variance structure, we will in the following present a result from optimality theory that allows to derive a solution which is valid irrespective of the present variance structure. 
The goal is to provide a comprehensive overview of this approach and its application to the problem at hand, either by explicitly presenting the relevant results along with their derivations or by referring to the corresponding results in the literature (see, e.g., \cite{pukelsheim.2006}), so that interested readers can find everything they need to fully follow the line of reasoning.\\

In the following, for an arbitrary subset $U \subset \mathbb{R}^p$, the notation $\mathcal{U} = \operatorname{conv}(U)$ is used  to denote its \textit{convex hull}, i.e., the set of all convex combinations of elements from $U$ which are of the form $\sum_i \kappa_i u_i$ with $\kappa_i \ge 0$, $\sum_i \kappa_i = 1$ and $u_i \in U$. Any such set $U$ will be referred to as a \textit{generating set} of $\mathcal U$. Further, let $\partial \mathcal{U}$ denote the \textit{boundary} of $\mathcal{U}$, i.e., the set of points in $\mathcal{U}$ for which every neighborhood contains points both inside and outside of $\mathcal{U}$.\\

\begin{figure*}[t!]   


\centering
       \includegraphics[width=1\textwidth]{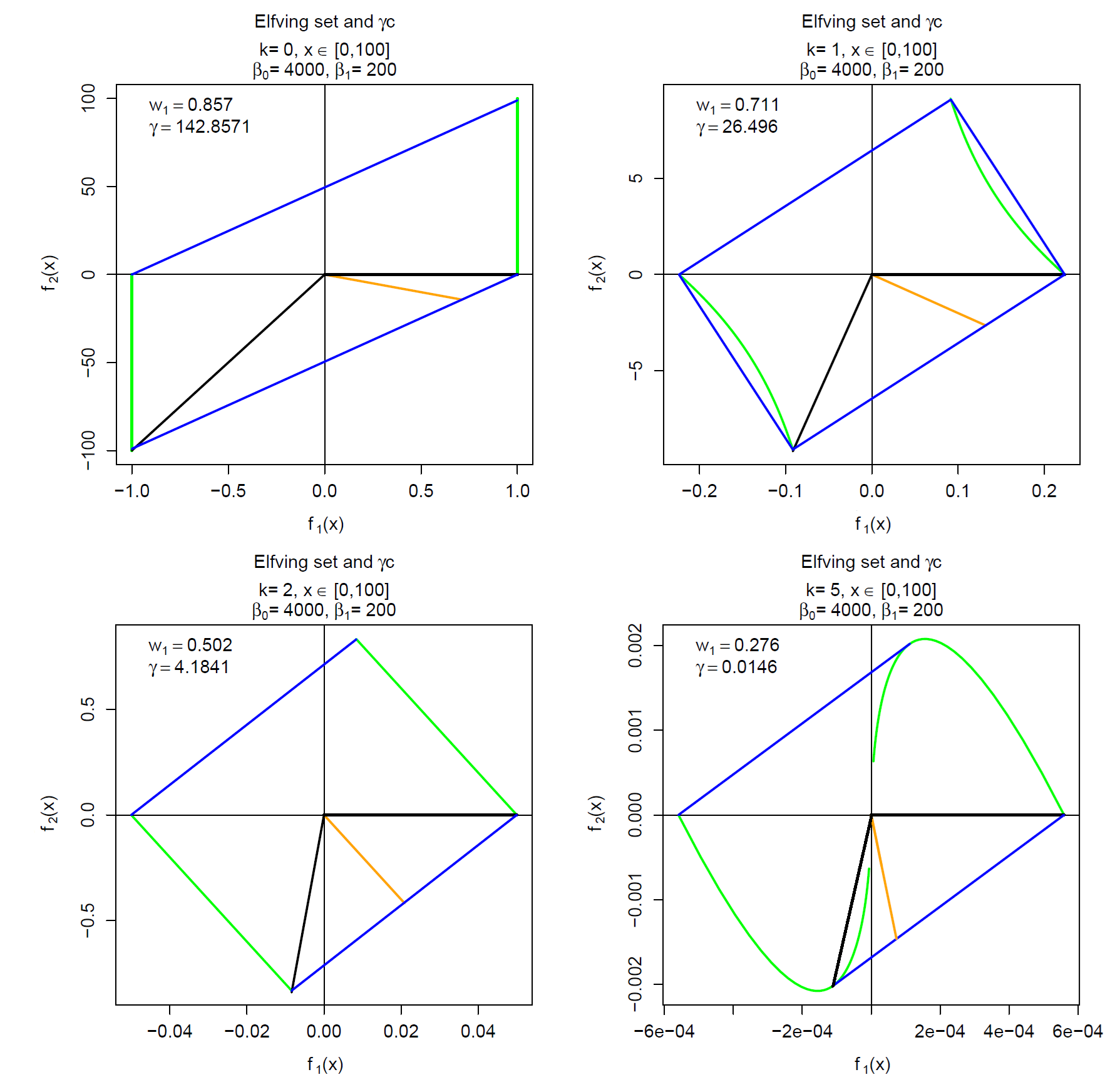}%
       
     \caption{Elfving sets for several different forms of variance growth as determined by $k$ (w.l.o.g. $V_0$ is set to $0$ in all 4 cases). Green lines indicate $f(x)$ and $-f(x), x \in [0,r]$, blue lines indicate the boundaries of the Elfving set where not given by the green lines, bold black lines indicate $f(0)$ and $-f(x^*)$ respectively and orange lines indicate $\gamma \bm c$.}

\label{Elvset}

\end{figure*}

\vspace*{3\baselineskip}

\begin{newthm}[Elfving, 1952, \cite{elfving.1952}, \cite{dette.1993}]
\label{Elfthm}

Given a compact design region $\mathcal{X}$ and a regression function $f(\cdot)$, the Elfving set $\mathcal{E}$ is defined as
\vspace*{-.25\baselineskip}
\[
\mathcal{E} = \operatorname{conv} \left( \{ f(x) \mid x \in \mathcal{X} \} \cup \{ -f(x) \mid x \in \mathcal{X} \} \right)
\]
with generating set $E=\{\pm f(x) \mid x \in \mathcal{X}\}$.\\

A design $\xi^* = \{(x_i,\kappa_i)\}_{i=1}^k$ is $\bm c$\textbf{-optimal} for estimating $\bm{c}^T \bm{\beta}$ if and only if there exist $\gamma > 0$ and signs $e_i \in \{-1,1\}$ such that
\vspace*{-.5\baselineskip}
\[
\gamma \bm{c} = \sum_{i=1}^k \kappa_i e_i f(x_i)
\quad \text{and} \quad
\gamma \bm{c} \in \partial \mathcal{E}.
\]

\end{newthm}

Elfving's Theorem provides a necessary and sufficient \textit{geometric characterization} of $c$-optimality in terms of the convex hull of the regression function $f(\cdot)$ (see Fig. \ref{Elvset}). Provided that the regression vectors \(f(x)\), \(x\in[0,r]\), span \(\mathbb R^2\), $\mathcal{E}$ has nonempty interior in which it  contains the origin. For standard addition this is always the case, since the vectors $f(0)=(1/\sqrt{v(0)},0)$ and $f(x)=(1/\sqrt{v(x)},x/\sqrt{v(x)}), x > 0,$ are not collinear. With respect to the existence of a $c$-optimal design, this implies that such a design exists if there exists a scalar $\gamma > 0$ such that $\gamma \bm c \in \partial \mathcal{E}$. Any boundary point $\gamma \bm c$ admits a representation as a convex combination of elements in $E$, which induces a design $\xi^*$ that minimizes $\bm c^T M(\xi)^{-1}\bm c$. In particular, Elfving's Theorem allows one to identify optimal support points and corresponding weights via a representation of a certain boundary point of $\mathcal{E}$. Questions of uniqueness, as well as the exact number of support points, however, generally require additional arguments beyond Elfving's result.\\

\section{Results}

\subsection{Elfving's Theorem: applicability and optimal solution} 
\label{subsec3.1}

\vspace*{1\baselineskip}

\paragraph{Standard addition and the applicability of Elfving's Theorem}\mbox{}\\

Minimizing the approximate variance of $\hat C_0$, as obtained in equation  \ref{varf} by the method of error propagation, is equivalent to solving a properly stated c-optimal design problem, since the approximation of the variance of $\hat C_0$, given by  equation \ref{varf} as

\[
\left(\frac{1}{\beta_1}\right)^2\sigma_{{\hat{\beta}}_0}^2+\left(\frac{\beta_0}{\beta_1^2}\right)^2\sigma_{{\hat{\beta}}_1}^2-\frac{2cov\left({\hat{\beta}}_0,{\hat{\beta}}_1\right)\beta_0}{\beta_1^3}
\]

\vspace*{.5\baselineskip}

can be written in terms of the variance of $\bm c^T\bm{\hat \beta}$ when choosing $\bm c$ appropriately. The proper $\bm{c}$ can be found by solving the equation

\begin{equation*}
\begin{aligned}
var(\bm c^T\bm{\hat \beta})=\bm c^T var(\bm{\hat \beta}) \bm c
&= \bm c^T \left(\begin{matrix}\sigma_{{\hat{\beta}}_0}^2&cov\left({\hat{\beta}}_0,{\hat{\beta}}_1\right)\\cov\left({\hat{\beta}}_0,{\hat{\beta}}_1\right)&\sigma_{{\hat{\beta}}_1}^2\\\end{matrix}\right) \bm c\\
&=\left(\frac{1}{\beta_1}\right)^2\sigma_{{\hat{\beta}}_0}^2+\left(\frac{\beta_0}{\beta_1^2}\right)^2\sigma_{{\hat{\beta}}_1}^2-\frac{2cov\left({\hat{\beta}}_0,{\hat{\beta}}_1\right)\beta_0}{\beta_1^3}
\end{aligned}
\end{equation*}

\vspace*{.5\baselineskip}

with respect to $\bm c= (c_1,c_2)^T$, which yields 

\vspace*{-1\baselineskip}

\[
\bm c =
\begin{pmatrix}
1/\beta_1\\
-\beta_0/\beta_1^2
\end{pmatrix}
.
\]

This allows the use of Elfving's theorem.\\

\paragraph{Number of support points necessary}\mbox{}\\

\textit{Carath\'eodory's Theorem} \cite{tyrrell.1970} states that any point in the convex hull $ \mathcal{U}$ of a set $U \subset \mathbb{R}^p$ can be expressed as a convex combination of at most $p+1$ points from $U$. Since $p=2$ in the standard addition context, this implies that any boundary point of an Elfving set $\mathcal{E}$ corresponding to a certain standard addition problem admits a representation involving at most $2+1$ elements from its generating set $E=\{\pm f(x)\mid x\in[0,r]\}=f([0,r])\cup(-f([0,r]))$. Furthermore, $\mathcal{E}$ can be easily illustrated in this context, since the behavior of the variance of the measurement error is given by the smooth variance function $v(x)=V_0 + (\beta_1(C_0 + x))^k$. Therefore, $\{f(x)=(1/\sqrt{v(x)},x/\sqrt{v(x)})^T \mid x \in \mathcal{X}\}$ describes a smooth compact curve in $\mathbb{R}^2$. This curve is, depending on $k$, either convex or concave (see Fig. \ref{Elvset}  for $k=0,1,2$ (convex) and $5$ (concave)). Thus, the boundary of such an  Elfving set can only consist of line segments  joining points of $E$ and curve arcs contained in $E$. Since the coordinates of $f(x)$ are always non-negative for $x \geq 0$, the image $f([0,r])$ and its reflection at the origin, $-f[0,r]$, are always fully contained in the first and third quadrant. Hence, in the second and fourth quadrants, the boundary of \(\mathcal E\) is formed by exactly one line segment in each case, connecting the respective points of \(f([0,r])\) and \(-f([0,r])\). Consequently, an optimal design can either be constructed using a single support point if $\gamma \bm{c}$ coincides with a point in $\partial \mathcal{E} \cap E$, or two points if it lies on a line segment between two distinct points contained in $E$. Therefore, in this setting, the representation reduces in fact to at most $2$ support points. Due to the special structure of $ \mathcal{E}$ there are definitely two support points needed if $\bm c$ lies in the second or fourth quadrant of $ \mathbb{R}^2$ where no points of $E$ can be found. In the case of standard addition, $\gamma \bm c$ is always in the fourth quadrant if $\beta_i > 0, i=0,1,$ since $\bm c=(1/\beta_1,-\beta_0/\beta_1^2)$. This is equal to $C_0 > 0$ $(\Leftrightarrow \beta_0 > 0$) and a sensitivity  $>0$ of the analytical instrument used ($\Leftrightarrow \beta_1 > 0$). Note: the support points of the respective optimal design are therefore always the same, irrespective of $\beta_0$ and $\beta_1$ as long as $\beta_0, \beta_1 > 0$. A reasonable exception occurs only for the case where $\beta_1>0$ and $C_0=0$. The latter is equal to $\beta_0=0$ and thus $\gamma \bm c=f(0) \in E$. Therefore, the support of the corresponding optimal design consists only of $x_1=0$ which is in line with the intuition that the best possible result with respect to $\sigma_{\hat C_0}$ for the case $C_0=0$ can be achieved if only the unspiked sample is measured.\\

\paragraph{\boldmath Determination and uniqueness of the support points of $\xi^*$}\mbox{}\\

Unless otherwise stated, for all that follows, $p$ is restricted to $2$, i.e., $\mathbb{R}^p=\mathbb{R}^2$, since this is the relevant case for standard addition. To introduce the notion of a \textit{supporting line} \cite{pukelsheim.2006}, let \(M\) be an arbitrary set, \(L_T\) an arbitrary tangent to \(M\), and \(L_S\) an arbitrary supporting line of \(M\). Similar to $L_T$, $L_S$ touches $M$ at its boundary; however, unlike tangents, which may intersect the interior $\operatorname{int}(M)=M \setminus \partial M$ of the set away from a point of contact $x_i \in \partial M \cap L_i, i=T,S$, a supporting line never intersects the interior of the set and keeps therefore the entire set on one side. Moreover, supporting lines may exist at boundary points where no tangent is defined, for instance at the corners of a rectangle, where infinitely many supporting lines exist, whereas no tangent line is available. A line $L$ that touches the boundary of a set $M$ at $x_0 \in \partial M$ is said to \textit{locally support} $M$ at $x_0$ when there exists a ball $B_\epsilon(x_0)$, i.e., a ball with radius $\epsilon > 0$ arbitrary small centered around $x_0$, such that $L \cap \operatorname{int}(M) \cap B_\epsilon(x_0) = \emptyset$. If this holds also for $\epsilon$ arbitrary large, $L$ becomes a supporting line. Therefore, by definition, $L_S$ is locally supporting $M$ at every point of contact, but this is not necessarily true for $L_T$. Consider, for example, the inner boundary of a closed annulus in $\mathbb{R}^2$ where no such $B_\epsilon$ exists for any of its elements. In contrast to non-convex sets, every line $L$ locally supporting a convex set $C$ at an arbitrary point of contact, is automatically also a supporting line of $C$. If there existed points in $C$ lying on different sides of $L$, then $L$ would separate points of $C$ and hence intersect the interior of $C$ since, by convexity, for any two points in $C$ the line segment connecting them is entirely contained in $C$ thus contradicting the assumed local support.\\

\begin{figure*}[b!]   
  
\centering
       \includegraphics[width=.6\textwidth]{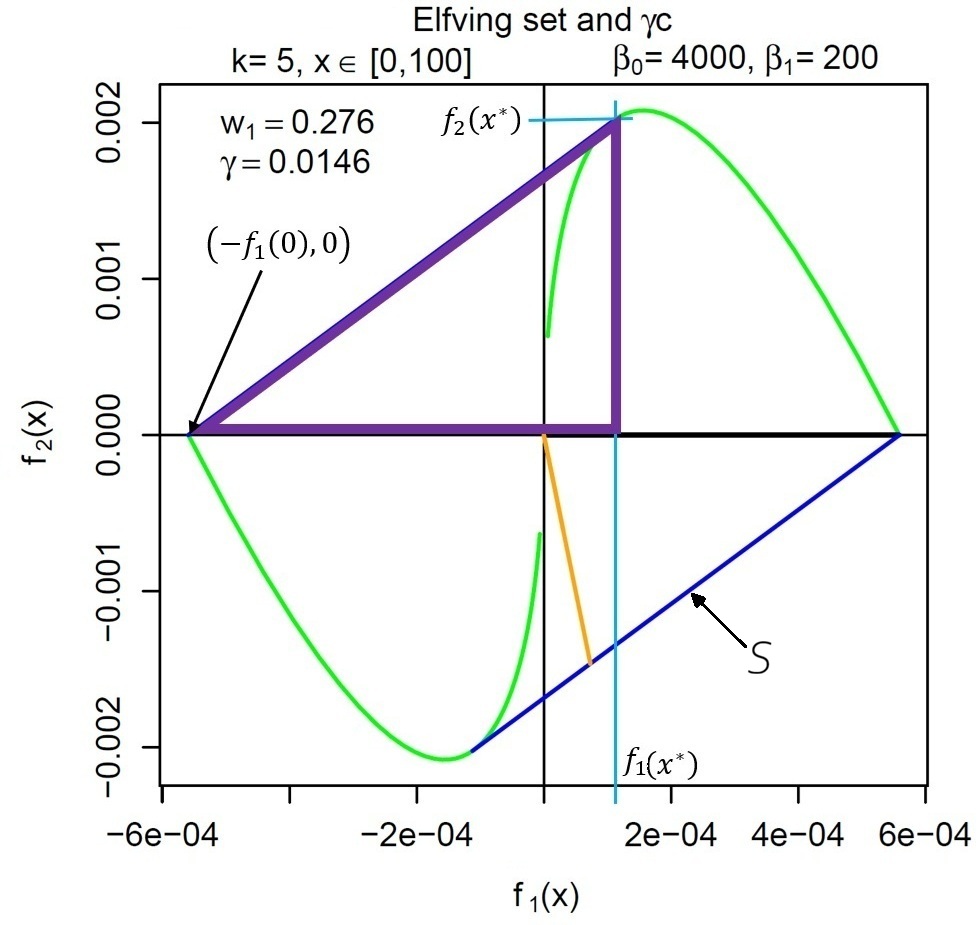}%
       
     \caption{Line segment tangent to $f([0,r])$, the green curve on the right, at $f(x^*)=(f_1(x^*),f_2(x^*))$, connecting $f(x^*)$ and $-f(0)$ (hypotenuse of the bold purple triangle).}

\label{TangLine}

\end{figure*}\

\vspace*{-1\baselineskip}

Applying the concept of supporting lines to the convex Elfing set $\mathcal E$ yields: Since, as discussed above, the vector $\bm{c}$ has a positive first and a negative second coordinate if $C_0 > 0$, the intersection between $\gamma\bm{c}$ and $\partial\mathcal{E}$ is in this case located in the fourth quadrant. It lies on the line segment connecting $-f(x^*)$ and $f(0)$, denoted $S$ in the following, with $x^* \leq r$ (see Fig. \ref{Elvset}). The point $x^* \in [0,r]$ is defined as the the smallest value such that $f(x^*)\in\ell\cap f([0,r]),$ where $\ell$ is obtained by clockwise rotating a line about $-f(0)$ until it first touches the curve $f([0,r])$ (the green curves on the right in the subfigures of Fig. \ref{Elvset}) without intersecting it; its exact value, however, depends on $k$ (see Fig. \ref{TangLine} for an illustration of the case $k=5$). Additional considerations show, that for the line $\ell$, passing through $-f(0)$ and $f(x^*)$, it holds that $\ell \cap \{\pm f(x) : x \in [0,r]\}=\{-f(0),f(x^*)\}$ for all $k\geq0$ (see the Appendix for a general discussion of this topic), i.e., that $f(x^*)$ is also the unique point of tangency of $\ell$ with $f([0,r])$. Therefore, it holds for the line  $L$ passing through $-f(x^*)$ and $f(0)$, that it is a supporting line of $\mathcal{E}$ with $L \cap \mathcal{E}=S$. Furthermore, $\mathcal{E}$ is by construction above $L$ for all $k \geq 0$, i.e., every point of $\mathcal{E}$ has second coordinate at least as large as the point on $L$ with the same first coordinate. Since $ f(0)$ and $- f(x^*)$ are the generating elements of $S$ they can be convexly combined such that $\gamma\bm{c}=\kappa_1f(0)+\kappa_2(-f(x^*))$ with $\kappa_1+\kappa_2=1$. Therefore, due to Elfving's Theorem, the design $\xi^*$, defined by the support points $\{0,x^*\}$ and the corresponding weights $\kappa_1$ and $\kappa_2$, is $c$-optimal. The uniqueness of the optimal design $\xi^*$ follows from the fact, that $\xi^*$ is optimal if and only if $\gamma \bm{c}$ can be represented by a convex combination of the elements of $\{\pm f(x_i) \mid x_i \in \text{supp}(\xi^*)\}$. This is possible only for a convex combination of $ f(0)$ and $- f(x^*)$ which is demonstrated in the following.  \\

\begin{lemma*}[Representation of the supporting line $L$ of $\mathcal{E} \subset \mathbb{R}^2$ passing through $f(0)$ and $-f(x^*)$]\label{lem}\mbox{}\\
\hypertarget{lem}{}

There exist $\bm z \in \mathbb{R}^2$ and $a \in \mathbb{R}$ such that for all $\bm x \in L, \bm z^T\bm x=a$ and for all $\bm x \in \mathcal{E} \setminus L, \bm z^T\bm x > a$. 

\begin{proof}

Like for every affine line in $\mathbb{R}^2$ there exist $a, b \in \mathbb{R}$ such that it holds for all $\bm x=(x_1,x_2)^T \in L$ that
\[
x_2=a+bx_1.
\]
From this it follows, that with $\bm z = (-b,1)^T$ it holds that 
\[
\bm z^T\bm x=a \ \text{for all} \ \bm x \in L. 
\] 
Since all $\bm x=(x_1,x_2)^T \in \mathcal{E} \setminus L$ are above $L$, i.e., $x'_2 < x_2$ for $(x_1,x'_2)^T \in L$, it holds that for every $\bm x \in \mathcal{E} \setminus L$, there exists $a' > a$ such that $x_2=a'+bx_1$ and therefore 
\[
\bm z^T\bm x=-bx_1+x_2=-bx_1+a'+bx_1=a' > a.   
\]
\end{proof}
\end{lemma*}

\begin{proposition*}[Uniqueness of $\xi^*$]\label{prop}\mbox{}\\
\hypertarget{prop}{}

Let $\mathcal{E} \subset \mathbb{R}^2$ be the Elfving set arising from the problem of minimizing the approximation of $\sigma^2_{\hat C_0}$ in the context of standard addition (expression \ref{varf}). Then the optimal approximate design $\xi^*$ resulting from the application of Elfving's Theorem is uniquely defined.

\begin{proof}

The Proof is subdivided into two parts: (1) uniqueness of support points $x_i$ and (2) determination of the respective proportions $\kappa_i$.\\ 

\textbf{(1) Uniqueness of support points}\\

This part is based on the lemma given above.\\

Let $L$ be a supporting line of $\mathcal{E}$ passing through $f(0), -f(x^*)$ and $\gamma \bm c$, and let $U \subset \mathcal{E}$. \\

We distinguish three cases:

\medskip
\textbf{(1.a)} Assume $U \cap L = \emptyset$. Then $\bm z^T \bm x > a$ for all $\bm x \in U$. For any $\bm x \in \operatorname{conv}(U)$,
\[
\bm  x = \sum_{i=1}^n \lambda_i \bm x_i, \quad \bm x_i \in U,\ \lambda_i \ge 0,\ \sum_{i=1}^n \lambda_i = 1,  
\]
we obtain
\[
\bm z^T \bm x = \sum_{i=1}^n \lambda_i \bm z^T \bm x_i > \sum_{i=1}^n \lambda_i a = a.
\]
Hence $\bm x \notin L$, and therefore
\[
\operatorname{conv}(U) \cap L = \emptyset.
\]

\medskip
\textbf{(1.b)} Assume $U \cap L = \{\bm u\}$. Let $\bm  x \in \operatorname{conv}(U) \cap L$ with representation
\[
\bm x = \sum_{i=1}^n \lambda_i \bm x_i, \quad \bm x_i \in U, \ \lambda_i > 0,\ \sum_{i=1}^n \lambda_i = 1.
\]
If there exists $i$ with $\bm x_i \neq \bm u$, then $\bm z^T \bm x_i > a$, while $\bm z^T \bm u = a$, hence
\[
\bm z^T \bm x > a,
\]
a contradiction. Thus $\bm x = \bm u$, and
\[
\operatorname{conv}(U) \cap L = \{\bm u\}.
\]

\medskip
\textbf{(1.c)} Assume $U$ contains two distinct points $\bm u_1, \bm u_2 \in L$, and let
\[
\bm x = \lambda \bm u_1 + (1-\lambda)\bm u_2, \quad \lambda \in (0,1),
\]
be a point on the line segment connecting $\bm u_1$ and $\bm u_2$. Suppose that
\[
\bm x = \sum_{i=1}^n \lambda_i \bm x_i, \quad \bm x_i \in U,\;\lambda_i > 0,\;\sum_{i=1}^n \lambda_i = 1,
\]
with at least one $\bm x_j \notin L$. Then $\bm z^T \bm x_j > a$, while $\bm z^T \bm x_i = a$ for all $\bm x_i \in L$, and hence
\[
\bm z^T \bm x > a,
\]
contradicting $\bm x \in L$. Therefore, any such representation uses only points of $L$. Since $\bm x$ lies on the line segment between $\bm u_1$ and $\bm u_2$, it follows that
\[
\bm x = \lambda \bm u_1 + (1-\lambda)\bm u_2
\]
is the unique convex combination representation using points of $U$, and no representation involving points from $U \setminus L$ with strictly positive weights is possible.\\

It follows from \textbf{(1.a)} through \textbf{(1.c)} that the support of the optimal design $\xi^*$ for standard addition is always given by the  points $\{0,x^*\}$, which implies, in particular, that an optimal design for standard addition is always a two-point design.\\

\textbf{(2) Determination of the respective proportions} $\bm \kappa_i$\\

The optimal proportions of measurements, $\kappa_1$ and $\kappa_2$, to be allocated to the support points $0$ and $x^*$ of the optimal design $\xi^*$ follow from solving the equation

\begin{center}
$\gamma\bm{c}=\kappa_1f(0)+\kappa_2(-f(x^*))$ such that $\kappa_1+\kappa_2=1$.\\ 
\end{center}

With

\vspace*{-1\baselineskip} 

\[
f(x)=
\begin{pmatrix}
1/\sqrt{(v(x))}\\
x/\sqrt{(v(x))}
\end{pmatrix}
=(f_1(x),f_2(x))^T
\]

one gets

\vspace*{-.5\baselineskip}

\[
\gamma\bm{c}=\gamma(c_1,c_2)^T=\kappa_1f(0)+\kappa_2(-f(x^*))=(\kappa_1f_1(0)+\kappa_2(-f_1(x^*)) , \kappa_1f_2(0)+\kappa_2(-f_2(x^*)))
\]

which yields

\vspace*{-.5\baselineskip}

\begin{equation}\label{syseq}
    \begin{aligned}
\gamma c_1 &=\kappa_1f_1\left(0\right)-\left(1-\kappa_1\right)f_1(x^*)\\
\gamma c_2 &=\kappa_1f_2\left(0\right)-\left(1-\kappa_1\right)f_2(x^*)
    \end{aligned}
    \quad \Longleftrightarrow \quad
    \begin{aligned}
\gamma c_1 - \kappa_1(f_1(0)+f_1(x^*)) &=-f_1(x^*)\\
\gamma c_2 - \kappa_1(f_2(0)+f_2(x^*)) &=-f_2(x^*)
    \end{aligned}
\end{equation}

\vspace*{.5\baselineskip}

The solution to this system of linear equations is unique if 

\[
c_2(f_1(0)+f_1(x^*)) - c_1(f_2(0)+f_2(x^*))  \neq 0.
\]

Since

\vspace*{-1\baselineskip} 

\[
\bm c =(c_1,c_2)^T= (1/\beta_1,-\beta_0/\beta_1^2)^T \ \text{with} \ \beta_0 \geq 0, \beta_1 > 0 \ \text{and} \ f_i(0) \geq 0, f_i(x^*) > 0, i=1,2,
\]

\vspace*{.5\baselineskip} 

$c_1 > 0$ and $c_2 \leq 0$ guarantees that 

\[
c_2(f_1(0)+f_1(x^*)) - c_1(f_2(0)+f_2(x^*)) < 0
\]

always holds. Therefore, solving the system of equations given in \ref{syseq} with respect to $\kappa_1$ yields the unique solution

\vspace*{-1\baselineskip} 

\begin{equation}\label{eqkap}
\kappa_1=\frac{f_2\left(x^*\right)c_1-f_1\left(x^*\right)c_2}{c_1(f_2\left(0\right)+f_2\left(x^*\right))-c_2(f_1\left(0\right)+f_1\left(x^*\right))}.
\end{equation}

\end{proof}
\end{proposition*}

\vspace*{.5\baselineskip} 


From the \hyperlink{prop}{proposition} and the additional considerations concerning \(\mathcal{E}\) given above (e.g., regarding \(x^*\)), it follows that the optimal design $\xi^*$ for standard addition is, for all $k \geq 0$, a uniquely defined two-point design. This means $\xi^*$ can be unambiguously determined with respect to the support and the proportions of observations. The optimal support, $\text{supp}(\xi^*)$, is equal to $\{0,x^*\}$ and the optimal proportions of observations, $\kappa^*_1$ and $\kappa^*_2$, can be calculated by using formula \ref{eqkap} given in the proposition. $x^*$ is equal to $r$ for $k \leq 2,$ but can be smaller than $r$ for $k > 2$, see e.g., example 4 in Fig. \ref{Elvset} (in this example, $k=5$ and the optimal support is approximately $\{0,18.4\}$ while $r=100$). To get an exact design $\xi_n$ based on $\xi^*$, the optimal number of observations for the unspiked sample, \(n_1^*\), is obtained by appropriately rounding \(\kappa_1^* n\), while the corresponding number for the spiked sample is given by $n_2^* = n - n_1^*$. Since only the rounding of \(\kappa_1^* n\) is necessary, determination of  $n_1^*$ and  $n_2^*$ can be easily achieved without resorting to elaborate rounding strategies.\\


\vspace*{-.5\baselineskip} 

\paragraph{A few additional comments on the calculation of the optimal two-point design}\mbox{}\\

For calculating $\kappa_1$ by applying formula \ref{eqkap} given in the \hyperlink{prop}{proposition} the following quantities are needed:

\vspace*{-1\baselineskip} 

\[
\begin{pmatrix}
f_1(x)\\
f_2(x)
\end{pmatrix}
=
\begin{pmatrix}
1/\sqrt{(v(x))}\\
x/\sqrt{(v(x))}
\end{pmatrix}
\] 

\vspace*{.5\baselineskip}

with $v(x)=V_0+(\beta_1(C_0+x))^k$ with $\beta_1 > 0$ and $k,V_0, C_0 \geq 0$. In addition one needs 

\vspace*{-.5\baselineskip} 

\[
\begin{pmatrix}
c_1\\
c_2
\end{pmatrix}
=
\begin{pmatrix}
1/\beta_1\\
-\beta_0/\beta_1^2
\end{pmatrix}.
\]

\vspace*{.5\baselineskip} 

Unfortunately this means that knowledge of $\beta_0, \beta_1, V_0$ and $k$ is required to calculate the optimal approximate design. No role is played by the number of total measurements $n$ and a constant proportionality factor $\sigma^2 > 0$ which are only necessary to calculate the actual magnitude of the variance. However, since all quantities needed are normally not fully known, it is in practice often not possible to determine the best possible design beforehand. Thus, it is investigated in the following, how well a two-point design with $50:50$ allocation of the measurements performs in comparison to a multi-point design using equidistant spiked concentrations and the optimal two-point design.\\


\subsection{Bias and precision: Significance of the design (including range)}
\label{ResBiasSD}

For a maximum permissible spiked concentration $r$, the design is defined in the Introduction as the vector

\vspace*{-1.5\baselineskip}

\[
\bm{x}=(x_{11},...,x_{1n_1},x_{21},...,x_{2n_2},...,x_{n_e1},...,x_{n_en_{n_e}}) 
\]

with spiked concentrations $x_1,...,x_{n_e} \in [0,r]$ each repeated $n_1,...,n_{n_e}$ times indicating $n_i$ independent measurements allocated at the spiked concentrations $x_i$, i.e., $x_i=x_{i1}=x_{i2}=...=x_{in_i}, i=1,...,n_e$. $\sum n_i = n$ is the total number of measurements taken and for each fixed $n$ different designs can yield, as detailed above, different bias and  precision of the estimator $\hat{C}_0$. It is further shown in \autoref{subsec3.1} that the variance optimal design is a two-point design and that the corresponding optimal allocation of measurements can be determined by formula \ref{eqkap}. In the following, the allocation of measurements for optimal two-point designs for several different settings will be given. Moreover, as outlined above, if $k>2$ (see Fig. \ref{TangLine} for $k=5$), the optimal spiked concentration $x^*$ can be smaller than $r$, thus rendering the application of a two-point design in such cases potentially problematic if $x^*$ is not at least moderately well-known. Hence, to obtain a clearer picture of the complex interaction between the various factors, namely the impact of the number of spiked concentrations, allocation of measurements, permissible range for spiking and variance behavior on the estimator performance, the results of various simulations and calculations are presented below. These results do not only pertain to the variance, but also to the bias of $\hat{C}_0$ to demonstrate how the latter relates to the former.\\

The chosen settings are, for all results shown in this subsection, determined by the parameters $\beta_0=4000$ and $\beta_1=200$ in combination with varying types of variance behavior. The latter is modeled according to equation \ref{varfunc}, i.e., the variance with respect to the spiked concentration $x$ is given by

\vspace{-.6cm}

\[
\sigma^2(x)=\sigma^2v(x)=\sigma^2(V_0+(\beta_1(C_0+x))^k)=\sigma^2(V_0+y^k).
\]

\vspace{.2cm}

$\sigma$ represents a proportionality constant necessary for determining the magnitude of the variance and the exponent $k$ determines the type of variance increase, i.e., $k=0$ represents the homoscedastic case and $k>0$ the heteroscedastic case. For example, $k=2$ represents quadratic variance increase which is often equal to an (almost) constant RSD. $V_0$ is needed to determine the variance for a blank which is denoted in the following by $\sigma^2_0$. If $k>0$, $V_0$ is equal to $\sigma^2_0/\sigma^2$ because $\sigma^2_0=\sigma^2V_0$ since $(\beta_1\cdot 0)^k=0$.  For example, if $\sigma=0.03$ and $\sigma_0=400$, $V_0=400^2/0.03^2\approx 177777778$. If $k=0$, then $V_0$ can be simply set to zero since $(\beta_1\cdot 0)^0=1$.\\ 

The impact of different allocations of measurements for a two-point design ($x_1=0, \ x_2=r$) on precision and bias of $\hat C_0$  for $k=0,2$ is illustrated in Figs. \ref{SDsim} and \ref{Biassim}, respectively. These figures show results for the same settings and in each of it results with respect to simulations and approximations are overlaid. The approximations for variance and bias are gained by using Eqs. \ref{varf} and \ref{biasf}, respectively. In addition, the figures illustrate the influence of the choice of $r$ on the variance, the bias and the optimal allocation of measurements.\\

\begin{table}[h!]
\centering
\caption{Optimal values of $\kappa_1$, $\kappa_1 \cdot n$, and the optimally rounded integer values of $\kappa_1 \cdot n$ for two-point designs with $n=12$ under four different settings. Rounding was carried out subject to the constraint that neither $n_1$ nor $n_2$ is allowed to equal zero.}
\label{opt2p}
\small
\begin{tabular}{*{3}{c} | *{4}{c} | *{4}{c} | *{4}{c} }
 \hline
    \multicolumn{3}{c|}{} & 
     \multicolumn{4}{c|}{range} & 
    \multicolumn{4}{c|}{range} & 
    \multicolumn{4}{c}{range} \\ 
  \hline

    \multicolumn{3}{c|}{parameters} & 
   50 & 100 & 1000 & 10000  & 
   50 & 100 & 1000 & 10000  & 
  50 & 100 & 1000 & 10000   \\ 

  \hline

  $k$ & $\sigma$ & $\sigma_0$ & 
   \multicolumn{4}{c|}{optimal $\kappa_1$} & 
    \multicolumn{4}{c|}{$\kappa_1 \cdot n$} & 
    \multicolumn{4}{c}{$n_1=\kappa_1 \cdot n$ opt. round.} \\ 

 \hline
0 & 400 & 400 & 0.78 & 0.86 & 0.98 & 1.00 & 9.33 & 10.29 & 11.77 & 11.98 & 9 & 10 & 11 & 11 \\ 
  2 & 0.03 & 0 & 0.50 & 0.50 & 0.50 & 0.50 & 6.00 & 6.00 & 6.00 & 6.00 & 6 & 6 & 6 & 6 \\ 
  2 & 0.03 & 20 & 0.50 & 0.50 & 0.50 & 0.50 & 6.04 & 6.04 & 6.04 & 6.04 & 6 & 6 & 6 & 6 \\ 
  2 & 0.03 & 400 & 0.72 & 0.75 & 0.78 & 0.78 & 8.59 & 9.03 & 9.32 & 9.32 & 9 & 9 & 9 & 9 \\
   \hline
\end{tabular}
\end{table}

Tab. \ref{opt2p} contains results corresponding to Fig. \ref{SDsim} in that it presents the optimal number of measurements $n_1$  allocated to $x_1$ for $n=12$ for the same settings underlying the results presented in that figure. The optimal $n_1$ given in Tab. \ref{opt2p} are based on $\kappa_1$, the optimal proportion of measurements for the unspiked sample, gained by applying equation \ref{eqkap}. \\

\begin{figure*}[b!]   
  
\centering
       \includegraphics[width=1\textwidth]{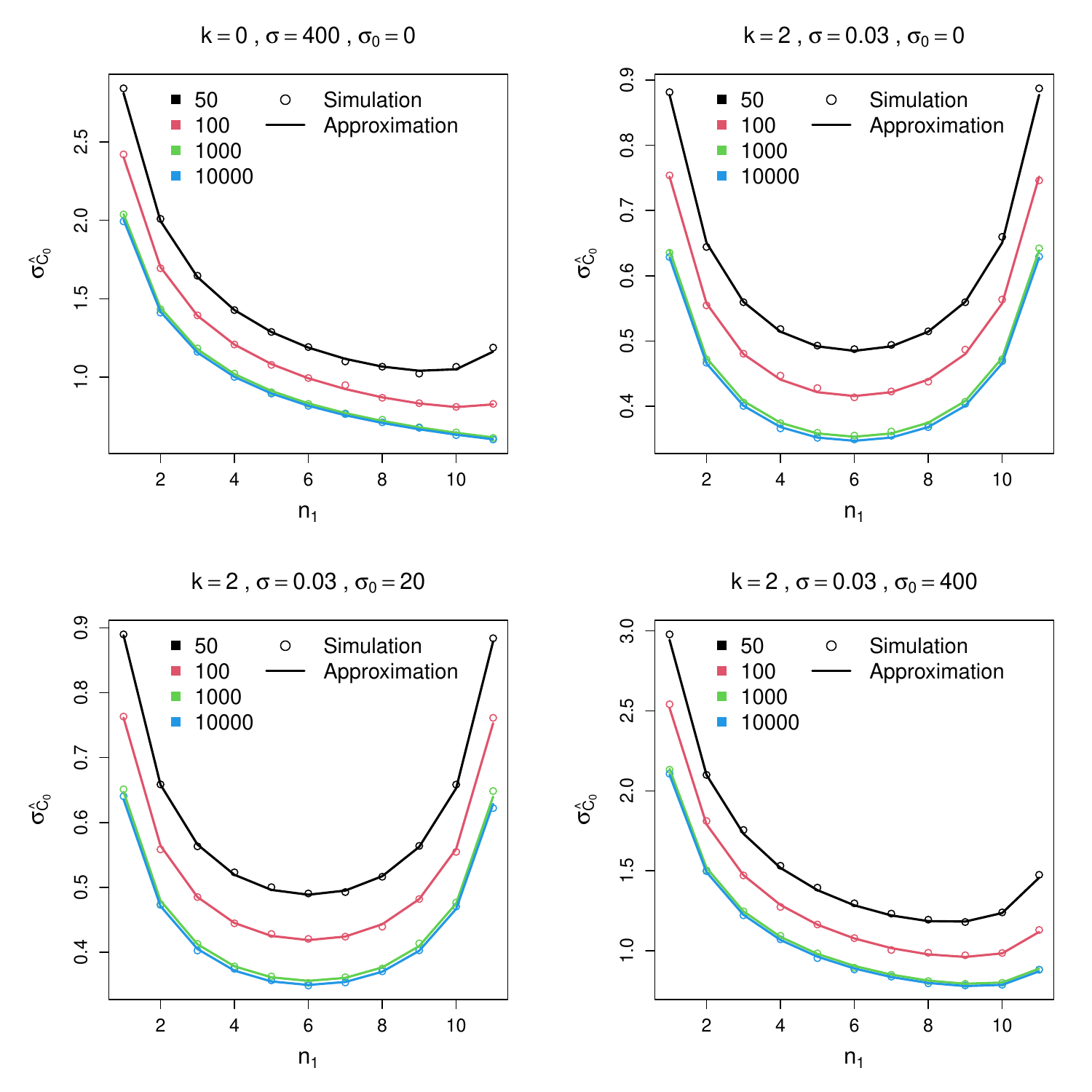}%
       
     \caption{Estimator precision with respect to measurement allocation for the two-point design for $n=12$. Each result is based on simulations using $10^4$ iterations. The numbers on the x-axis denote the numbers of independent measurements $n_1$ for $x_1=0$ ($n_2=n-n_1$ for $x_2=r$). $n_1$ is optimal, when $\sigma_{\hat{C}_0}$, given on the y-axis, is lowest.}

\label{SDsim}

\end{figure*}


\begin{figure*}[t!]   
  
\centering
       \includegraphics[width=1\textwidth]{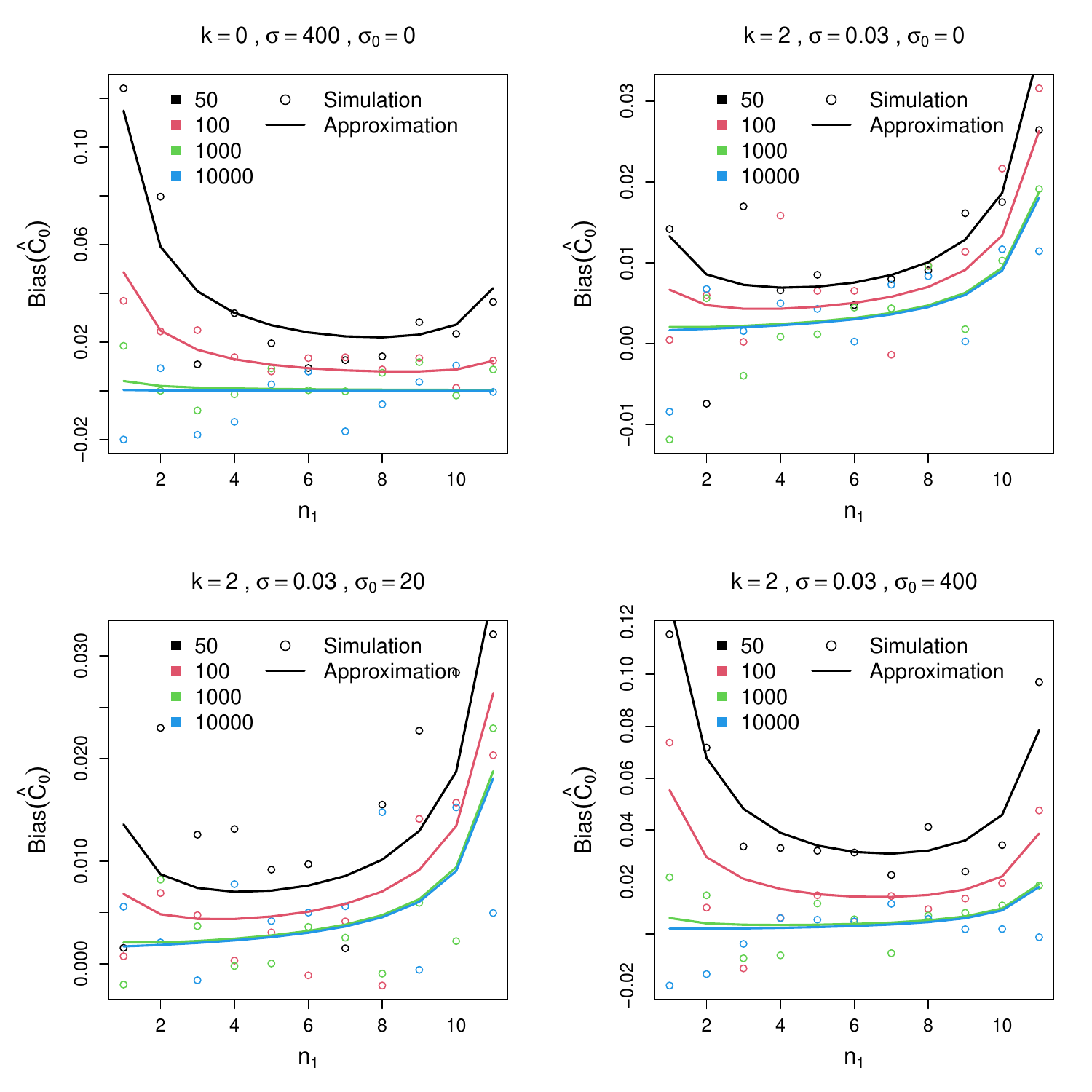}%
       
     \caption{Estimator bias with respect to measurement allocation for the two-point design for $n=12$, based on simulations using $10^4$ iterations each. The numbers on the x-axis denote the numbers of independent measurements $n_1$ for $x_1=0$ ($n_2=n-n_1$ for $x_2=r$).}

\label{Biassim}  

\end{figure*}

\begin{table}[h!]
\caption{The six different settings on which the results given in Tabs. \ref{SDxstar} and \ref{Biasxstar} are based.}
\label{Param}
\centering
\begin{tabular}{c|c c c l c c c}
Setting   & $\beta_0$ & $\beta_1$ & $V_0$ & $\sigma$ & $k$ & $r_{max}$ \\ \hline
1 &4000  & 200 & 0 &400   & 0 & 10$^4$ \\
2 &4000  & 200 &  0& 5 & 1 &  10$^4$\\
3 &4000  & 200 & 0 & 0.03  & 2 & 10$^4$  \\
4 &4000  & 200 &  0& 0.0003  &3  &10$^3$  \\
5 &4000  & 200 &  0& 0.000003 & 4 & 10$^2$ \\
6 &4000  & 200 &  0& 0.0000003 & 5 &10$^2$  
\end{tabular}
\end{table}

Tables \ref{SDxstar} and \ref{Biasxstar} give a detailed comparison of multi-point designs combined with unweighted and weighted regression and two-point designs with 50:50 allocation of measurements to optimal designs in terms of standard deviation and bias of $\hat C_0$. All considered multi-point designs consist of $n_e=4$ equidistantly spiked concentrations in $[0,r]$, with $x_1=0$ and $x_4=r$. The notion "weighted multi-point design" indicates in the following that the determination of $C_0$ based on the data from applying the respective multi-point design is done by using weighted regression with weights inversely proportional to the actual variances. The 50:50 two-point design consists of the spiked concentrations $x_1=0$ and $x_2=r$, with equal allocation of measurements, i.e., $n_1=n_2=n/2$. Finally, the optimal two-point design consists of the spiked concentrations $x_1=0$ and $x_2=x^* \leq r$ ($x^*$ printed in bold in column $r$ of Tabs. \ref{SDxstar} and \ref{Biasxstar}), with an allocation of measurements based on the respective result obtained by applying equation \ref{eqkap} which determines the optimal measurement proportion $\kappa_1$ for $x_1$. Hence, $n_1$ (given in column ${n_1}$ of Tabs. \ref{SDxstar} and \ref{Biasxstar}) is equal to the optimally rounded value of $\kappa_1n$ and $n_2=n-n_1$. These designs are compared with respect to six different settings. All settings share the same values for $\beta_0, \beta_1$ and $V_0$, but differ with respect to variance behavior - see the respective parameters given in Tab. \ref{Param}. In addition, four different values of the range are considered within each setting. Note that $\sigma$, and for $k>2$ also $r$, are chosen sufficiently small to avoid unrealistic values of $\sigma^2(x), x \leq r$, which would render the results devoid of any practical significance (see $\sigma$ and $r_{max}$ in Tab. \ref{Param}).\\

Moreover, as previously mentioned, when applying an optimal design, a guaranteed improvement in the precision of $\hat C_0$ resulting from an increase in the highest spiked concentration $x_{n_e}$, can only be expected if $k \leq 2$. This follows from the fact, that $x^*$ equals $r$ whenever $k \leq 2$, but not in general if $k>2$. In the latter case, there exists an upper bound $x^*_{max}$ for $x^*$, i.e., $x^*=r$ if $r \leq x^*_{max}$ and remains equal to $x^*_{max}$ if $r \geq x^*_{max}$. $x^*_{max}$ depends on the parameters describing the variance behavior and can be easily found by numerically solving the following equation (see Appendix):

\vspace*{-.2cm}

\begin{equation}
x^*_{max} = \operatorname*{arg\,max}_{x \in [0,r]} s_L(x) \quad \text{with} \quad  s_L(x)=\frac{f_2(x)}{f_1(x)+f_1(0)}.
\end{equation}

Therefore, since for the two-point design using a 50:50 allocation $x_2$ is always chosen equal to $r$, $x_2$ equals in this case $x^*$ irrespective of $r$ only if $k \leq 2$, whereas for $k>2$ this holds only if $r \leq x^*_{max}$. In contrast, the optimal design always sets   $x_2$ equal to $x^*$ and therefore $x_2=x^*_{max}$ for all $r \geq x^*_{max}$. In addition, the measurement allocation is optimal for the 50:50 allocation only if $k=2$ and $\sigma_0$ is rather small, whereas the optimal design of course always uses the optimal allocation.\\

The influence of the range for the multi-point and the 50:50 two-point design is straightforward to verify mathematically for the homoscedastic case, but considerably more complex in the heteroscedastic setting. Since the approximation formulas for bias and variance are linear combinations of $\sigma^2_{\hat \beta_0}, \sigma^2_{\hat \beta_1}$ and $cov(\hat \beta_0,\hat \beta_1)$ it can be tried to investigate the behavior of the different approaches by investigating how an increase of the range influences these quantities (see Supplement). For the homoscedastic case it can be shown, that an increase of the highest spiked concentration will not affect the variance of the estimator for the intercept, $\sigma^2_{\hat{\beta}_0}$, while that for the slope, $\sigma^2_{\hat{\beta}_1}$, decreases. In addition, also the covariance $cov(\hat{\beta}_0,\hat{\beta}_1)$ decreases with increasing range. In summary, it can therefore be concluded that an increase in $r$ leads to a decrease of $\sigma^2_{\hat C_0}$ in the homoscedastic case. Heteroscedasticity, however, significantly complicates matters, and analyzing the effect of an increasing range on the estimator variance based on the variables involved is no longer as straightforward as in the homoscedastic case. Nevertheless, a simple asymptotic analysis (see the

\begin{table}[H]
\centering
\caption{Comparison of unweighted (mult) and weighted multi-point designs (mult.w) to two-point designs with 50:50 allocation of measurements (50:50) and optimal two-point designs (opt) with  respect to standard deviation. sim...results gained by simulations, apx...approximate results gained by applying equation \ref{varf}. $x^*_{max}$ is, for $k>2$, given in column $r$ (printed in bold) and the optimal number of measurements allocated to $x_1$ can be found in column $n_1$.}
\label{SDxstar}
\begin{tabular}{rrrrrrrrrr}
  \hline
$k$ & $r$ & $n_1$ & $\sigma_{\hat C_0}^{mult.sim}$ & $\sigma_{\hat C_0}^{mult.w.sim}$ & $\sigma_{\hat C_0}^{mult.w.apx}$ & $\sigma_{\hat C_0}^{50:50 sim}$ & $\sigma_{\hat C_0}^{50:50 apx}$ & $\sigma_{\hat C_0}^{opt.sim}$ & $\sigma_{\hat C_0}^{opt.apx}$\\ 
  \hline
0 & 50.0 & 9 & 1.50 & 1.50 & 1.51 & 1.19 & 1.19 & 1.04 & 1.04 \\ 
  0 & 100.0 & 10 & 1.22 & 1.22 & 1.23 & 0.99 & 0.99 & 0.81 & 0.81 \\ 
  0 & 1000.0 & 11 & 0.99 & 0.99 & 0.99 & 0.83 & 0.83 & 0.62 & 0.62 \\ 
  0 & 10000.0 & 11 & 0.97 & 0.97 & 0.97 & 0.82 & 0.82 & 0.60 & 0.60 \\ 
  \hline
  1 & 50.0 & 8 & 1.52 & 1.40 & 1.39 & 1.04 & 1.02 & 0.99 & 0.98 \\ 
  1 & 100.0 & 9 & 1.36 & 1.15 & 1.15 & 0.84 & 0.84 & 0.77 & 0.77 \\ 
  1 & 1000.0 & 10 & 2.27 & 0.92 & 0.94 & 0.66 & 0.66 & 0.53 & 0.53 \\ 
  1 & 10000.0 & 11 & 6.52 & 0.91 & 0.92 & 0.65 & 0.65 & 0.48 & 0.48 \\ 
  \hline
  2 & 50.0 & 6 & 0.84 & 0.63 & 0.63 & 0.49 & 0.48 & 0.49 & 0.48 \\ 
  2 & 100.0 & 6 & 0.96 & 0.53 & 0.52 & 0.42 & 0.42 & 0.42 & 0.42 \\ 
  2 & 1000.0 & 6 & 4.66 & 0.42 & 0.41 & 0.35 & 0.35 & 0.35 & 0.35 \\ 
  2 & 10000.0 & 6 & 43.64 & 0.40 & 0.40 & 0.35 & 0.35 & 0.35 & 0.35 \\ 
  \hline
  3 & 30.0 & 5 & 0.79 & 0.58 & 0.57 & 0.48 & 0.48 & 0.47 & 0.47 \\ 
  3 & \textbf{60.0} & 4 & 0.96 & 0.49 & 0.48 & 0.47 & 0.46 & 0.44 & 0.44 \\ 
  3 & 667.0 & 4 & 11.17 & 0.64 & 0.64 & 0.97 & 0.95 & 0.44 & 0.44 \\ 
  3 & 1000.0 & 4 & 19.51 & 0.75 & 0.74 & 1.16 & 1.14 & 0.44 & 0.44 \\ 
  \hline
  4 & 14.0 & 4 & 0.68 & 0.57 & 0.57 & 0.47 & 0.47 & 0.46 & 0.46 \\ 
  4 & \textbf{28.3} & 4 & 0.73 & 0.45 & 0.45 & 0.44 & 0.44 & 0.40 & 0.41 \\ 
  4 & 67.0 & 4 & 1.29 & 0.43 & 0.43 & 0.57 & 0.57 & 0.40 & 0.41 \\ 
  4 & 100.0 & 4 & 1.94 & 0.47 & 0.47 & 0.72 & 0.72 & 0.40 & 0.41 \\ 
  \hline
  5 & 9.0 & 4 & 12.52 & 6.87 & 4.90 & 4.82 & 4.02 & 4.61 & 3.89 \\ 
  5 & \textbf{18.4} & 3 & 7.71 & 4.32 & 3.79 & 4.25 & 3.68 & 3.70 & 3.35 \\ 
  5 & 67.0 & 3 & 7050.79 & 5.14 & 4.16 & 90.51 & 7.34 & 3.73 & 3.35 \\ 
  5 & 100.0 & 3 & 26376444.62 & 8.66 & 5.14 & 1739.37 & 10.95 & 3.70 & 3.35 \\ 
   \hline
\end{tabular}
\end{table}


Supplement) still provides insight  to some extent into how the estimator for $C_0$ behaves. When no weighting is applied, increasing the range decreases $var(\hat{\beta}_1)$ for $k<2$ and should leave it unaffected for $k=2$. $var({\hat{\beta}}_0) $ increases with an increasing range for $k>0$ and $cov({\hat{\beta}}_0,{\hat{\beta}}_1)$ decreases for $k<1$ and should be  unaffected for $k=1$. Putting all this together gives no clear picture for $k < 2$ but suggests a clear increase of $\sigma^2_{\hat C_0}$ for $k \geq 2$. The behavior predicted by the asymptotic analysis is also confirmed by simulations, telling that an unrestricted extension of the range is advantageous in the unweighted case only for $k=0$. Interestingly, when proper weighting is applied, the asymptotic analysis yields the same results as in the unweighted case, but simulations show that the behavior of $var({\hat{\beta}}_0) $ is not as expected, instead it seems to be pretty stable which might be due to a stabilizing effect of $x_1=0$ (for $x_1 > 0$, simulations indeed show the expected behavior). However, when proper weighting is applied, simulations suggest that an unrestricted extension of the range could be advantageous up to $k = 2$.\\

\begin{table}[H]
\centering
\caption{Comparison of unweighted (mult) and weighted multi-point designs (mult.w) to two-point designs with 50:50 allocation of measurements (50:50) and optimal two-point designs (opt) with  respect to bias. sim...results gained by simulations, apx...approximate results gained by applying equation \ref{biasf}. $x^*_{max}$, calculated with respect to variance optimality of $\hat C_0$, is, for $k>2$, given in column $r$ (printed in bold) and the optimal number of measurements allocated to $x_1$ can be found in column $n_1$. }
\label{Biasxstar}
\footnotesize 
\begin{tabular}{ccc|c|cc|cc|cc}
 \hline
    \multicolumn{3}{c|}{Parameters} & 
    \multicolumn{7}{c}{$Bias(\hat C_0)$}\\
 \hline
$k$ & $r$ & $n_1$ & ${mult.sim}$ & ${mult.w.sim}$ & ${mult.w.apx}$ & ${50:50sim}$ & ${50:50apx}$ & ${opt.sim}$ & ${opt.apx}$\\
  \hline
0 & 50.0 & 9 & 0.053 & 0.05270 & 0.0432 & 0.0335 & 0.02400 & 0.0239 & 0.02311 \\ 
  0 & 100.0 & 10 & 0.023 & 0.02310 & 0.0168 & 0.0117 & 0.00933 & 0.0093 & 0.00880 \\ 
  0 & 1000.0 & 11 & 0.015 & 0.01491 & 0.0012 & 0.0087 & 0.00069 & 0.0065 & 0.00045 \\ 
  0 & 10000.0 & 11 & -0.004 & -0.00424 & 0.0001 & -0.0056 & 0.00007 & -0.0050 & 0.00004 \\ 
  \hline 
 1 & 50.0 & 8 & 0.058 & 0.04686 & 0.0432 & 0.0321 & 0.02333 & 0.0350 & 0.02625 \\ 
  1 & 100.0 & 9 & 0.044 & 0.04718 & 0.0189 & 0.0352 & 0.01000 & 0.0130 & 0.01333 \\ 
  1 & 1000.0 & 10 & -0.023 & -0.01084 & 0.0017 & -0.0089 & 0.00085 & -0.0080 & 0.00153 \\ 
  1 & 10000.0 & 11 & 0.120 & -0.00469 & 0.0002 & 0.0082 & 0.00008 & 0.0017 & 0.00027 \\ 
  \hline
  2 & 50.0 & 6 & 0.012 & 0.00016 & 0.0109 & 0.0056 & 0.00756 & 0.0056 & 0.00756 \\ 
  2 & 100.0 & 6 & 0.004 & -0.00009 & 0.0061 & -0.0009 & 0.00504 & -0.0009 & 0.00504 \\ 
  2 & 1000.0 & 6 & 0.028 & 0.00205 & 0.0024 & -0.0027 & 0.00318 & -0.0027 & 0.00318 \\ 
  2 & 10000.0 & 6 & 0.860 & 0.00208 & 0.0020 & 0.0070 & 0.00302 & 0.0070 & 0.00302 \\ 
  \hline
  3 & 30.0 & 5 & 0.010 & 0.01043 & 0.0119 & 0.0051 & 0.00967 & 0.0094 & 0.00874 \\ 
  3 & \textbf{60.0} & 4 & 0.026 & 0.00514 & 0.0080 & 0.0067 & 0.00907 & 0.0062 & 0.00720 \\ 
  3 & 667.0 & 4 & 0.720 & 0.02355 & 0.0178 & 0.0594 & 0.04377 & 0.0143 & 0.00720 \\ 
  3 & 1000.0 & 4 & 1.761 & 0.02163 & 0.0249 & 0.0684 & 0.06370 & 0.0049 & 0.00720 \\ 
  \hline
  4 & 14.0 & 4 & 0.018 & 0.01354 & 0.0134 & 0.0058 & 0.00985 & 0.0057 & 0.00863 \\ 
  4 & \textbf{28.3} & 4 & 0.023 & 0.00649 & 0.0082 & 0.0077 & 0.00873 & 0.0060 & 0.00698 \\ 
  4 & 67.0 & 4 & 0.034 & 0.00548 & 0.0078 & 0.0112 & 0.01550 & 0.0031 & 0.00698 \\ 
  4 & 100.0 & 4 & 0.074 & 0.01199 & 0.0097 & 0.0199 & 0.02500 & 0.0081 & 0.00698 \\ 
  \hline
  5 & 9.0 & 4 & 2.047 & 1.35766 & 1.0581 & 0.9052 & 0.74522 & 0.8216 & 0.66203 \\ 
  5 & \textbf{18.4} & 3 & 1.473 & 0.63799 & 0.6134 & 0.6784 & 0.63543 & 0.4826 & 0.48169 \\ 
  5 & 67.0 & 3 & -87.740 & 0.95001 & 0.8021 & 6.4338 & 2.67220 & 0.5252 & 0.48169 \\ 
  5 & 100.0 & 3 & 263760.960 & 1.69066 & 1.2632 & -10.6762 & 5.97658 & 0.5281 & 0.48169 \\ 
   \hline
\end{tabular}
\end{table}


\section{Discussion}

All results presented in this work regarding \hypertarget{start}{\hyperlink{vod}{variance optimal designs}} for standard addition are based on the assumption that the measurement variance $\sigma^2(x)$ is non-decreasing with increasing analyte concentration $x$. The assumed behavior is  described by $\sigma^2(x)=\sigma^2v(x)$,  with $\sigma > 0$ and variance function $v(x)=V_0+(\beta_1(C_0+x))^k$ as introduced by Eq. \ref{varfunc}. While the results presented in this work are valid for all $k \geq 0$, the particular important cases for practical applications are the case of constant error (homoscedasticity,  $k=0$) and the case of quadratic variance increase ($k=2$) or equivalently, the case of (approximately) constant relative standard error.\\ 

For what follows, in addition to the above assumption with respect to variance behavior, it is assumed that $C_0>0$. Then, the application of the results from optimality theory, i.e., \hyperref[Elfthm]{Elfving's Theorem} to the problem of design optimality for standard addition in \autoref{subsec3.1} yields the following unambiguous results (Note: designs are always compared based on the same number of observations, $n$):

\begin{itemize}
\item Irrespective of $k$, the variance optimal design is a two-point design, i.e., a design involving only two concentration levels. It is superior to all other possible alternative designs utilizing the same number of observations, regardless of how many spiked concentrations they use, and not only to other two-point designs.
\item The two support points of the optimal design, i.e., the concentrations $x_1$ and $x_2$ used for spiking, are uniquely defined, with $x_1=0$ being always the unique valid choice.
\item The optimal choice for $x_2$, denoted $x^*$ throughout this work, depends particularly on $k$ and $r$, with $r$ denoting the highest permissible concentration for spiking throughout this work.
\item If the variance grows at most quadratically with growing analyte concentration ($k \leq 2$), $x^*$ has always to be chosen equal to $r$.
\item If the variance grows more rapidly, i.e., $k>2$, the optimal $x^*$ might be smaller than $r$. Determining $x^*$ if $k>2$ can be done by proceeding as discussed in the Appendix. 
\item The optimal allocation $n_1$ and $n_2$ of a total of $n$ measurements to $x_1$ and $x_2$, respectively, with $n_1+n_2=n$, can be determined from Eq.~\ref{eqkap}.
\item In general, the optimal allocation does not correspond to an even distribution, i.e., $n_1=n_2$ cannot generally be assumed to hold.
\item Fortunately, in the particularly important case of quadratic variance growth, if the blank variance is small compared to the variance at $x_1$, the optimal allocation is an even distribution of measurements.
\end{itemize}

In contrast to the statements given above, the variance optimal design for a sample that contains no analyte, i.e., $C_0=0$, is a one-point design. This means that in this case it is, as also in line with intuition, optimal to measure only the unspiked sample. This result, however, only concerns the estimator variance and is of course not relevant to the question of whether the actual concentration can be inferred from such a result. For more information on this topic, see the issue of LOD and LOQ determination \cite{eurachem2025ffp}, but accurate determination of these quantities becomes particularly difficult in heteroscedastic settings when applying standard addition.\\

Since the intercept $\beta_0$ and the slope $\beta_1$ are held constant for all simulations and calculations presented in \autoref{ResBiasSD}, the results vary only with respect to $r$ and the variance behavior (the parameters appearing in the calculation of $\sigma^2(x)$). The results gained by simulations are based on $10^4$ iterations each whereas  the optimal proportion of observations $\kappa_1$ for $x_1$ and the approximate results for bias and variance of $\hat C_0$ are gained by applying Eqs. \ref{eqkap}, \ref{varf} and \ref{biasf}, respectively. The results regarding design optimality with respect to estimator precision gained by calculations given in Tab. \ref{opt2p} are exactly reproduced by the simulation results and approximations presented in Fig. \ref{SDsim}. These results, together with the results presented in Tab. \ref{SDxstar} are inline with the points stated above. In addition, by looking at Fig. \ref{SDsim} and into Tab. \ref{SDxstar}, it can be seen that the differences between the values for $\sigma^2_{\hat C_0}$ gained by simulation and the corresponding approximate values gained by applying equation \ref{varf}, are negligibly small up to $k=4$. This indicates that equation \ref{varf} provides reliable approximations which renders the application of time consuming simulations unnecessary, if the variance increase is not too rapid.\\

\paragraph{The results given in \autoref{ResBiasSD} in detail:} 

Tab. \ref{opt2p} shows that the optimal allocation of measurements, calculated by applying Eq. \ref{eqkap}, is only for $k=2$ in combination with small values of $\sigma_0$ equal to an even distribution but for all other settings deviates significantly from equality of $n_1$ and $n_2$. Fig. \ref{SDsim} shows how the precision of $\hat C_0$ depends on the allocation of measurements in the case of a two-point design for the same settings on which also the results in Tab. \ref{opt2p} are based. The results presented in this figure are gained by simulations and approximations, which, as already stated do not differ significantly. The optimal allocations that can be inferred from these figures are identical to the results given in Tab. \ref{opt2p}, i.e., the results based on applying formula \ref{eqkap} which is derived by applying optimality theory are identical to those derived by simulations, thus demonstrating the validity of the theoretical results. Furthermore, one can see that the precision increases with increasing $r$ since $x^*$ is always equal to $r$  for $k\leq 2$.\\

Tab. \ref{SDxstar} presents results regarding the comparison of several different designs for $n=12$. These designs are equidistantly spaced multi-point designs with $n_e=4$ with and without subsequent weighted regression as well as two-point designs with 50:50 allocation and optimal two-point designs. In the case of the two-point design with 50:50 allocation, $x_2$ is always set equal to $r$ and in the case of the optimal design, $x_2$ is equal to $r$ if $r \leq x^*_{max}$ and equal to $x^*_{max}$ if $r > x^*_{max}$. If $\bm{k=0}$, no weighting is necessary and thus both multi-point designs are equal. In this case, it can be observed that  $\sigma_{\hat C_0}^{mult}=\sigma_{\hat C_0}^{mult.w}>\sigma_{\hat C_0}^{50:50}>\sigma_{\hat C_0}^{opt}$ irrespective of $r$ and that increasing $r$ is beneficial in all cases. If $\bm{0 < k \leq 2}$, application of a multi-point design without subsequent weighted regression causes massive precision loss when arbitrarily increasing $r$, whereas application of weighted regression resolves this issue and yields an unconditional improvement of the result by increasing $r$. However, the 50:50 two-point design is even better than the weighted multi-point design, but of course outperformed by the optimal design, i.e., $\sigma_{\hat C_0}^{mult}>\sigma_{\hat C_0}^{mult.w}>\sigma_{\hat C_0}^{50:50}>\sigma_{\hat C_0}^{opt}$ irrespective of $r$. Increasing $r$ is beneficial for all but the unweighted multi-point approach, but does not lead to a uniform improvement. Instead, the gain in precision decreases progressively with increasing $r$, such that extending the range beyond a certain concentration results only in marginal further improvement. In summary, for $k \leq 2$, the following can be stated unambiguously: if the optimal design is not available (see also the \hyperlink{5050D}{remark} at the end of this section), apply a 50:50 two-point design and a large, but not necessarily excessively large, $x_2$ relative to the expected $C_0$. For example, for the specific case of $C_0=20, x_2=10^3$ already appears sufficiently large, since increasing $x_2$ to $10^4$ yields only a marginal further improvement. The case $\bm{k>2}$: in all considered cases, i.e., $k=3,4,5$, the values for $x^*_{max}$ remain below $100$, with respective values of approximately  $60, 28.3$ and $18.4$. All approaches (albeit with restrictions for the unweighted multi-point design) perform reasonable as long as $x_{n_e}$ is not too far from $x^*_{max}$. As expected, it holds as in case $k\leq2$ that, if $x_{n_e}=x^*$ for all designs, $\sigma_{\hat C_0}^{mult}>\sigma_{\hat C_0}^{mult.w}>\sigma_{\hat C_0}^{50:50}>\sigma_{\hat C_0}^{opt}$. Furthermore, multi-point designs without subsequent weighted regression show catastrophic results when $r$ is greatly enlarged. For the other three designs, when increasing $r$ beyond $x^*_{max}$, the following behavior can be observed: while the optimal design does not change any more since $x^*=x^*_{max}$ if $r>x^*_{max}$ and thus $\sigma_{\hat C_0}^{opt}$ stays constant, $\sigma_{\hat C_0}^{50:50}$ increases and gets at some point worse than $\sigma_{\hat C_0}^{mult.w}$. The latter seems to have its optimum beyond $x^*_{max}$ (which is of course not below $\sigma_{\hat C_0}^{opt}$ if the optimal design uses $x^*_{max}$) and to be more stable when increasing $r$ than the 50:50 two-point design. However, increasing $r$ beyond a certain point has for all, but the optimal design, a negative impact on estimator precision, thus, contrary to the case $k\leq2$, an arbitrary increase of the range is not beneficial if $k>2$. Therefore, in the latter case, if $x^*_{max}$ is not at least approximately known, one should aim to progressively improve the determination of the concentration by iteratively adapting the measurement design in subsequent steps.\\


Fig. \ref{Biassim} and Tab. \ref{Biasxstar} show results regarding the bias of $\hat C_0$ based on simulations and approximations which correspond to those shown above for the precision of the estimator. Within both, the figure and the table, the simulation and approximation results are paired. Due to the often rather small values of the bias, the fluctuation of this quantity seen in simulations is relatively high which is the reason for the discrepancies between the simulated values and the calculated approximate values based on equation \ref{biasf}. Nevertheless, it can be seen that simulations and approximations yield quantities of similar magnitude. However, the results of the approximations as well as the simulations indicate, that the variance optimal design is normally not identical to the optimal design with respect to the  bias of $\hat C_0$, but that the differences with respect to bias for these two designs are in general relatively moderate. Since, also with respect to the range applied, the bias shows a similar behavior as the variance, the statements made for the latter can also be applied more or less analogously to the former.\\


\hypertarget{5050D}{Finally}, it should once again be emphasized that, since determination of the optimal design needs knowledge of several, normally unknown quantities, it is often not possible to apply the optimal design in practice. Nevertheless, at least if $k \leq 2$, a two-point design with 50:50 allocation of measurements is expected to yield significantly better results, even in cases where it is not optimal, than common multi-point designs. However, if the variance function deviates substantially from this assumption, i.e., if $k$ is significantly larger than $2$, an iterative strategy may be appropriate in practice, particularly when concentrations within a narrow range are to be routinely determined. Initially, a sufficient number of measurements must be taken to obtain not only a sufficiently accurate estimate of the expected value of $C_0$ and the sensitivity, i.e., $\beta_1$, but in particular a reasonably accurate characterization of the variance function in order to satisfactorily determine approximate values of $x^*$ and $\kappa_1$. Thus, when the quantities necessary are known to a certain degree, it should be possible to get a design that's close to the optimal design which might be especially of value when standard addition is repeatedly applied.\\


\noindent \textbf{Supplementary Information}  Link to supplementary material.

\section{Declarations}

\noindent \textbf{Competing Interests}\\ 
The authors declare that they have no conflicts of interest for this work.\\

\noindent \textbf{Funding}\\ 
No funding was received for conducting this study.\\

\bibliography{c-opt-ref}

@article{valverde2023,
  title={Analytical quality by design using a D-optimal design and parallel factor analysis in an automatic solid phase extraction system coupled to liquid chromatography. Determination of nine PAHs in coffee samples},
  author={Valverde-Som, Luc{\'\i}a and Arce, MM and Sarabia, LA and Ortiz, MC},
  journal={Chemometrics and Intelligent Laboratory Systems},
  volume={243},
  pages={105008},
  year={2023},
  publisher={Elsevier}
}

@article{wang2025,
  title={Multi-response optimization and validation analysis in the detection of acetochlor and butachlor by HPLC based on D-optimal design methodology},
  author={Wang, Yu and Yang, Yang and Zhang, Si and Chi, Chao and Wang, Bohan and Su, Junfeng},
  journal={Journal of Chromatography A},
  volume={1754},
  pages={466025},
  year={2025},
  publisher={Elsevier}
}

@book{eurachem2025ffp,
  editor    = {H. Cantwell},
  title     = {The Fitness for Purpose of Analytical Methods: A Laboratory Guide to Method Validation and Related Topics},
  edition   = {3rd},
  year      = {2025},
  publisher = {Eurachem},
  url       = {https://www.eurachem.org},
  note      = {Available from www.eurachem.org}
}

@article{kiefer.1959,
  title={Optimum experimental designs},
  author={Kiefer, Jack},
  journal={Journal of the Royal Statistical Society: Series B (Methodological)},
  volume={21},
  number={2},
  pages={272--304},
  year={1959},
  publisher={Wiley Online Library}
}

@book{larson.2005,
  title={Calculus},
  author={Larson, Ron and Hostetler, Robert P and Edwards, Bruce},
  year={2005},
  publisher={Cengage Learning}
}

@book{rao.1973,
  title={Linear statistical inference and its applications},
  author={Rao, Calyampudi Radhakrishna},
  volume={2},
  year={1973},
  publisher={Wiley New York}
}

@article{tyrrell.1970,
  title={Convex analysis},
  author={Tyrrell Rockafellar, R},
  journal={Princeton mathematical series},
  volume={28},
  year={1970},
  publisher={Princeton university press Princeton, NJ, USA}
}

@book{nocedal.2006,
  title={Numerical optimization},
  author={Nocedal, Jorge and Wright, Stephen J},
  year={2006},
  publisher={Springer}
}

@article{karush.1939,
  title={Minima of functions of several variables with inequalities as side constraints},
  author={Karush, William},
  journal={M. Sc. Dissertation. Dept. of Mathematics, Univ. of Chicago},
  year={1939}
}

@article{kitsos.2010,
  title={An optimal calibration design for pH meters},
  author={Kitsos, Christos P and Kolovos, Konstantinos G},
  journal={Instrumentation Science and Technology},
  volume={38},
  number={6},
  pages={436--447},
  year={2010},
  publisher={Taylor \& Francis}
}

@article{dette.1993,
  title={Elfving's theorem for D-optimality},
  author={Dette, Holger},
  journal={The Annals of Statistics},
  pages={753--766},
  year={1993},
  publisher={JSTOR}
}

@article{elfving.1952,
  title={Optimum allocation in linear regression theory},
  author={Elfving, Gustav},
  journal={The Annals of Mathematical Statistics},
  pages={255--262},
  year={1952},
  publisher={JSTOR}
}

@article{ratzlaff.1979,
  title={Optimizing precision in standard addition measurement},
  author={Ratzlaff, Kenneth L},
  journal={Analytical Chemistry},
  volume={51},
  number={2},
  pages={232--235},
  year={1979},
  publisher={ACS Publications}
}

@book{montgomery.2017,
  title={Design and analysis of experiments},
  author={Montgomery, Douglas C},
  year={2017},
  publisher={John wiley \& sons}
}

@book{pukelsheim.2006,
  title={Optimal design of experiments},
  author={Pukelsheim, Friedrich},
  year={2006},
  publisher={SIAM}
}

@book{Montgomery.2013,
 author = {Montgomery, Douglas C. and Peck, Elizabeth A. and Vining, G. Geoffrey},
 year = {2013},
 title = {Introduction to Linear Regression Analysis},
 url = {https://ebookcentral.proquest.com/lib/kxp/detail.action?docID=1211887},
 address = {s.l.},
 edition = {5. Aufl.},
 publisher = {Wiley},
 isbn = {9780470542811},
 series = {Wiley Series in Probability and Statistics}
}

@article{Aitken.1935,
 author = {Aitken, A. C.},  
 year = {1935},
 title = {On Least Squares and Linear Combinations of Observations},
 journal = {Proc. R. Soc. Edinb.},
 doi = {10.1017/s0370164600014346}
}

@article{Gossler.2025,
 author = {G{\"o}ssler, Gerhard and Hofer, Vera and Goessler, Walter},
 year = {2025},
 title = {Evaluation of four different standard addition approaches with respect to trueness and precision},
 journal = {Anal. Bioanal. Chem.},
 doi = {10.1007/s00216-024-05725-8}
}

@article{Ellison.2008,
 author = {Ellison, Stephen L. R. and Thompson, Michael},
 year = {2008},
 title = {Standard additions: myth and reality},
 pages = {992--997},
 volume = {133},
 number = {8},
 journal = {The Analyst},
 doi = {10.1039/b717660k}
}

@article{Goncalves.2016,
 author = {Goncalves, Daniel A. and Jones, Bradley T. and Donati, George L.},
 year = {2016},
 title = {The reversed-axis method to estimate precision in standard additions analysis},
 pages = {155--158},
 volume = {124},
 issn = {0026265X},
 journal = {Microchem. J.},
 doi = {10.1016/j.microc.2015.08.006}
}

@article{Franke.1978,
 author = {Franke, J. P. and de Zeeuw, R. A. and Hakkert, R.},
 year = {1978},
 title = {Evaluation and optimization of the standard addition method for absorption spectrometry and anodic stripping voltammetry},
 pages = {1374--1380},
 volume = {50},
 number = {9},
 issn = {0003-2700},
 journal = {Anal. Chem.},
 doi = {10.1021/ac50031a045}
}


\section{Appendix}

\subsection{Determination and uniqueness of $x^*$}


It is assumed that $V_0 \ge 0, \quad \beta_1 > 0, \quad C_0 \ge 0, \quad k \ge 0 \ \text{and} \ x \in [0,r] = \mathcal{X}$. From this it follows, that, for $x \in \mathcal{X}$, the variance function $v(x)=V_0 + (\beta_1(C_0+x))^k$ is $\geq 0$ and monotonically increasing. To avoid a variance equal to zero, at least either $V_0$ or $C_0$ must be greater than zero. Furthermore, the regression function $f(x)$ is given by 

\vspace{-.2cm}

\[
f(x)=
\begin{pmatrix}
1/\sqrt{(v(x))}\\
x/\sqrt{(v(x))}
\end{pmatrix}
=(f_1(x),f_2(x))^T.
\]

To investigate the behavior of $f(\cdot)$ the following derivatives are needed:\\ 

$v'(x) = \beta_1^k \, k (C_0+x)^{k-1}, \ f_1'(x) = -\frac{1}{2} v(x)^{-3/2}v'(x)$ and $f_2'(x) = v(x)^{-1/2} - \frac{x}{2} v(x)^{-3/2} v'(x)$.\\ 

An extremum of the curve parametrized by $f(x)$ at $x \in (0,r)$ is only possible if

\[
f_2'(x) = 0 \ \Leftrightarrow \ 0=V_0 + \beta_1^k (C_0+x)^{k-1}
\left[C_0 + x\left(1-\frac{k}{2}\right)\right].
\]


The expression on the right-hand side can only be equal to zero if 
\[
C_0 + x\left(1-\frac{k}{2}\right) \leq 0. 
\]
For a closer examination, it is now necessary to distinguish between the following cases: $k < 2, k=2$ and $k > 2$. For $k<2$ it follows that $1-\frac{k}{2} > 0$ and therefore $f_2'(x) > 0$ for all $x \in (0,r)$, i.e., it is strictly increasing and has no critical points. For $k=2, 1 - k/2 = 0$, and the condition reduces to $V_0 + \beta_1^2 (C_0+x) C_0 = 0$. Since all terms are nonnegative, this holds if and only if $V_0 = 0 \quad \text{and} \quad C_0 = 0$. In this case, $f_2(x) = x/(\beta_1 x) = 1/\beta_1$, hence $f'_2(x) = 0 \quad \forall x > 0$. Therefore, for $k=2$, $f_2$ is either  constant or increasing. For $k > 2, 1 - k/2 < 0$ and $C_0 + x(1-k/2)$ is $< 0$ for all $x > C_0/(k/2-1)$. Therefore, if at least $V_0$ or $C_0$ are unequal to zero, $f'_2(0) > 0$ and

\[
\lim_{x \to \infty} \left(V_0 + \beta_1^k (C_0+x)^{k-1} \left[C_0 + x\left(1-\frac{k}{2}\right)\right]\right) = -\infty. 
\]

Hence, for $k > 2$, $f_2$ has exactly one critical point $x_m$ in $[0,\infty)$, which is a global maximum. Depending on $k, V_0$ and $C_0$, $x_m$ is either $\leq r$ or  $> r$.\\


The above considerations imply for $k \leq 2$ that it holds for the line $L$ passing through $-f(0)$ and $f(x^*)$ (see Fig. \ref{TangLine}) that $L \cap f([0,r]) = {f(x^*)}$, i.e., that the point of tangency of $L$ with $f([0,r])$ is unique.\\ 


For $k>2$ additional considerations are needed: The slope of the tangent to the curve is given by (see, e.g., \cite{larson.2005})

\vspace{-.5cm}

 \[
s_T(x) = \frac{f_2'(x)}{f_1'(x)},
\]

with

\vspace{-.5cm}

\[
{s'}_T\left(x\right)=1-\frac{2}{k}+\frac{V_0\left(k-1\right)\ }{k\left(\beta_1\left(C_0+x\right)\right)^k}>0 \ \text{for} \ k>2.
\]

\vspace{.2cm}

Hence, since $s'_T(x)>0 \ \text{for all } x \in [0,r]$, the tangent slope increases monotonically with $x$ and the curve continuously bends to the left and is thus strictly convex which implies that tangent lines touching the curve at different points cannot be parallel. Thus, also for $k>2$, $x^*$ is the unique point of tangency of $L$ with $f([0,r])$.\\  

Given the geometric properties of this problem, $x^*_{max}$ can be easily found by numerically solving

\vspace{-.2cm}

\begin{equation}
\label{xstardet}
x^*_{max} = \operatorname*{arg\,max}_{x \in [0,r]} s_L(x) 
\end{equation}

with

\vspace{-.5cm}

\begin{equation}
s_L(x)=\frac{f_2(x)}{f_1(x)+f_1(0)}
\end{equation}

\vspace{.2cm}

giving the slope of the line $L$ passing through $-f(0)$ and $f(x)$, $x \in [0,r]$ (sought-after is thus $x^* \in [0,r]$ such that $s_T(x^*)=s_L(x^*)$).\\

For example 4 ($k=5$), shown in Fig. \ref{Elvset}, $x^*$ is $\approx 18.4$ (in comparison $x_m \approx 13.33$).

\end{document}